\documentclass[twocolumn,preprintnumbers,aps,superscriptaddress]{revtex4}

\usepackage{graphicx}
\usepackage{amsmath}
\usepackage{amssymb}
\usepackage{amsthm}
\usepackage{color}

\newcommand{\be}{\begin{equation}}
\newcommand{\ee}{\end{equation}}
\newcommand{\ba}{\begin{array}}
\newcommand{\ea}{\end{array}}
\newcommand{\bea}{\begin{eqnarray}}
\newcommand{\eea}{\end{eqnarray}}

\newcommand{\calC}{{\cal C }}

\newcommand{\calD}{{\cal D }}

\newcommand{\calS}{{\cal S }}
\newcommand{\calG}{{\cal G }}

\newcommand{\ZZ}{\mathbb{Z}}

\newcommand{\la}{\langle}
\newcommand{\ra}{\rangle}

\newcommand{\nn}{\nonumber}

\newtheorem{lemma}{Lemma}

\begin{document}

\title{Subsystem surface codes with three-qubit check operators}

\author{Sergey \surname{Bravyi}}
\affiliation{IBM Watson Research Center, Yorktown Heights NY 10598 (USA)}

\author{Guillaume \surname{Duclos-Cianci}}
\affiliation{D\'epartment de Physique, Universit\'e de Sherbrooke, Sherbrooke, Qu\'ebec, J1K 2R1 (Canada)}

\author{David Poulin}
\affiliation{D\'epartment de Physique, Universit\'e de Sherbrooke, Sherbrooke, Qu\'ebec, J1K 2R1 (Canada)}

\author{Martin \surname{Suchara}}
\affiliation{Computer Science Division, UC  Berkeley, Berkeley CA 94720 (USA) }

\begin{abstract}
We propose a simplified version of the Kitaev's surface code
in which error correction  requires only three-qubit parity measurements
for Pauli operators $XXX$ and $ZZZ$.
The new code belongs to the class
of subsystem stabilizer codes. It inherits many favorable properties of the
standard surface code such as encoding
of multiple logical qubits on a planar lattice with punctured holes,
efficient decoding by either minimum-weight matching or renormalization group methods, and high error threshold.
The new subsystem surface code (SSC)
 gives rise to an exactly solvable Hamiltonian with
$3$-qubit interactions, topologically ordered ground state, and a
constant energy gap.
We construct a local unitary transformation mapping the SSC Hamiltonian to the one of the ordinary surface code  thus showing that the two Hamiltonians
belong to the same topological class.
We describe error correction
protocols for the SSC and
determine its error thresholds under several
natural error models.
In particular, we show that the SSC has error threshold
approximately $0.6\%$ for the standard circuit-based
error model studied in the literature. We also consider a
model in which
three-qubit parity operators can be measured directly. We show that the SSC
has error threshold approximately $0.97\%$  in this setting.
\end{abstract}

\maketitle

\section{Introduction}
Quantum error correcting codes are vital ingredients
in all scalable quantum computing architectures proposed so far.
By actively monitoring and correcting errors,
the encoded quantum states can be protected from noise up to any desired
precision  provided that
the error rate of elementary quantum operations is
below certain constant value known as the
error threshold.

Topological codes such as
the surface code family~\cite{Kitaev03,BK:surface,Dennis01}
have received considerable attention lately due to their
several attractive features. First, the
quantum hardware envisioned in the surface code architecture
consists of a 2D array of qubits with controlled nearest-neighbor interactions
and a local readout.  In principle, it can be implemented
using the Josephson junction qubits
technology~\cite{divincenzo:JJarchitecture}.
Surface codes feature an error threshold of at least  $1\%$~\cite{WFH:highthresh}
which is one of the highest thresholds among all studied codes.
Secondly,  encoded Clifford gates such as the CNOT gate can be implemented efficiently  by the code deformation method
\cite{BMD:codedef, RH:cluster2D,Fowler08}
which requires only a mild overhead in space and time.
The error rate of
encoded gates decreases exponentially with the lattice size~\cite{RHG:threshold}.
Thirdly, the surface codes can be decoded efficiently using
Edmonds's minimum weight matching algorithm~\cite{Dennis01,Fowler12}
or renormalisation group methods~\cite{Poulin09,DP10a1,BH11a}.

Although the surface code is among the best code candidates,
a promising direction for improvements has been recently identified by Bombin~\cite{bombin:topsub}
who proposed topological subsystem codes~\cite{BMD:topo}.
A subsystem code~\cite{Bacon06,Poulin05} can be viewed as a regular stabilizer code in which one or several logical
qubits do not encode any information. The presence of unused logical qubits, known as {\em gauge qubits},
simplifies  eigenvalue  measurements of multi-qubit stabilizers---such
as the plaquette and star operators of the surface code---which are
required for error correction. Consider as an example the simplest $4$-qubit code
with two stabilizers  $S^X=X^{\otimes 4}$ and $S^Z=Z^{\otimes 4}$. It encodes two qubits
with logical Pauli operators $\overline{X}_L=X_1X_2$, $\overline{Z}_L=Z_1 Z_3$
and $\overline{X}_G=X_1X_3$, $\overline{Z}_G=Z_1 Z_2$. If only the first logical qubit
is used to encode information, the syndrome (eigenvalue) of $S^X$ can be determined
indirectly by measuring
eigenvalues of the unused logical operators $\overline{X}_G$ and $S^X \overline{X}_G =X_2 X_4$.
Multiplying the measured eigenvalues together yields the desired eigenvalue of $S^X$.
The syndrome of $S^Z$ is determined similarly by
measuring eigenvalues of the unused logical operators
$\overline{Z}_G$ and $S^Z\overline{Z}_G=Z_3 Z_4$
followed by  multiplication of the outcomes.
 Hence the full syndrome extraction
requires only two-qubit parity measurements and simple classical post-processing. The unused logical operators
that need to be measured in order to extract the syndrome of all stabilizers
are usually referred to as {\em gauge generators}, see~\cite{Poulin05}
for the general theory of subsystem codes.

A  simplified syndrome readout offered by subsystem codes has its own costs. In any practical settings, eigenvalues
of individual gauge generators can only be  measured with a finite accuracy.
As one multiplies together measured  eigenvalues, errors tend to accumulate rendering the inferred
syndrome bit unreliable. This strongly limits the class of candidate subsystem codes
 for the topological fault-tolerant
architecture. First, a suitable code must have local gauge generators, ideally,
$2$- or $3$-qubit Pauli operators acting on nearest-neighbor qubits.
This ensures that  the syndrome readout requires only local measurements.
To avoid accumulation of measurement errors, a suitable code must also have
low-weight stabilizers. More precisely, each stabilizer must be composed of
only a few gauge generators.
 The latter requirement leaves out many interesting families of codes, such as the 2D
 Bacon-Shor codes~\cite{Bacon06}
 and random 2D subsystem codes discovered in~\cite{Bravyi11}.
In contrast, subsystem color codes found in~\cite{bombin:topsub} have $2$-qubit gauge generators while stabilizers
act on either $6$ or $18$ qubits.
These codes were shown to have a constant error
threshold of at least $2\%$ under depolarizing noise assuming noiseless syndrome readout~\cite{Suchara11,Bombin11}.
 Unfortunately, subsystem
color codes do not inherit favorable properties of the surface code such as encoding of
multiple logical qubits~\cite{bombin:topsub} on a planar lattice
required for the code deformation method.

In the present paper we propose a subsystem  version of the standard surface code
on the regular square lattice. Each plaquette of the lattice carries one gauge qubit
and a pair of weight-$6$ stabilizers of type $X^{\otimes 6}$ and $Z^{\otimes 6}$,
see Fig.~\ref{fig:code_lattice_toric}.
The code  has $3$-qubit gauge generators of type $XXX$ and $ZZZ$ which makes it
suitable for architectures where direct $3$-qubit parity measurements in the $X$- and $Z$-basis
are available. A promising proposal for implementing $3$-qubit parity measurements
in Josephson junction qubits has been recently made
by DiVincenzo and Solgun~\cite{DiVincenzo12}.
By analogy with the surface code, the new subsystem surface code (SSC)
can encode multiple logical qubits on a planar lattice
with punctured holes.  We describe error correction
protocols for the SSC and
determine its error thresholds under several
natural error models.
First, we study the so called ``code capacity'' model
where each qubit is subject to independent bit-flip and phase-flip errors
with rate $p$, while syndrome measurements are noiseless.
By relating the optimal error correction to the
phase transition in the  random-bond Ising model on the honeycomb
lattice, we show that the threshold error rate is $p_0\approx 7\%$.
Secondly, we study the  circuit-based error model
where the syndrome readout is simulated by noisy quantum gates,
single-qubit measurements, and ancilla preparations.
Each operation can fail with a probability $p$, see Section~\ref{sec:circuit} for details.
Monte Carlo simulation suggests that the
SSC has error threshold
$p_c \approx 0.6\%$ for the circuit-based error model.
Finally, we consider a model in which
$3$-qubit parity can be measured directly. We show that the SSC
has error threshold approximately $0.97\%$ for this direct parity measurement model.

The new code also gives rise to an exactly solvable Hamiltonian with $3$-qubit
interactions  which has  a topologically ordered ground state
and whose excitations are non-interacting abelian  anyons.
The exact solvability of the model stems from a peculiar commutativity
structure of the model's Hamiltonian. We show that the set of all relevant
$3$-qubit interactions
can be partitioned into small clusters such that interactions from different
clusters pairwise commute.
In contrast, it was shown recently by Aharonov and Eldar~\cite{Aharonov11}
 that topological order cannot
be realized by $3$-local Hamiltonians
in which {\em all} interactions pairwise commute.
We also construct a local unitary transformation $U$ that maps the $6$-qubit stabilizers
of the SSC to the plaquette and star operators of the Kitaev's surface code on the square lattice. The gauge generators $XXX$ and $ZZZ$ are mapped to single-qubit
$X$ and $Z$ operators. Loosely speaking, the map $U$ decouples gauge qubits
from the surface code qubits.

The rest of the paper is organized as follows.
Sections~\ref{sec:STC},\ref{sec:TQO}
introduce a subsystem version of the toric code
with $3$-qubit gauge generators, the corresponding exactly
solvable Hamiltonian and discuss its connection to the
ordinary toric code.
Extension to the planar geometry is given in Section~\ref{sec:SSC}
that defines  subsystem surface codes.
The concept of a virtual lattice
which is crucial for understanding our error correction protocols
is introduced in Section~\ref{sec:ec2D}. This section also discusses error correction
for the idealized settings when syndrome readout is noiseless.
Error correction protocols for the circuit-based
syndrome readout and numerical simulations are presented in Section~\ref{sec:circuit}.
Finally, Section~\ref{sec:toy} focuses on a model in which
direct $3$-qubit parity measurement are available.

\section{A subsystem  toric code }
\label{sec:STC}

We begin by introducing a subsystem version of the toric code that
encodes two logical qubits~\cite{Kitaev03}.
Extension to planar lattices with a boundary will be described in the next section.
The code is defined on the regular square lattice of size $L\times L$
with periodic boundary conditions. It contains $L^2$ vertices,
$2L^2$ edges, and $L^2$ plaquettes. We place a qubit
at every vertex and at the center of every edge of the lattice.
Hence there are in total $n=3L^2$  code qubits.
\begin{figure}[h]
\centerline{\includegraphics[height=7cm]{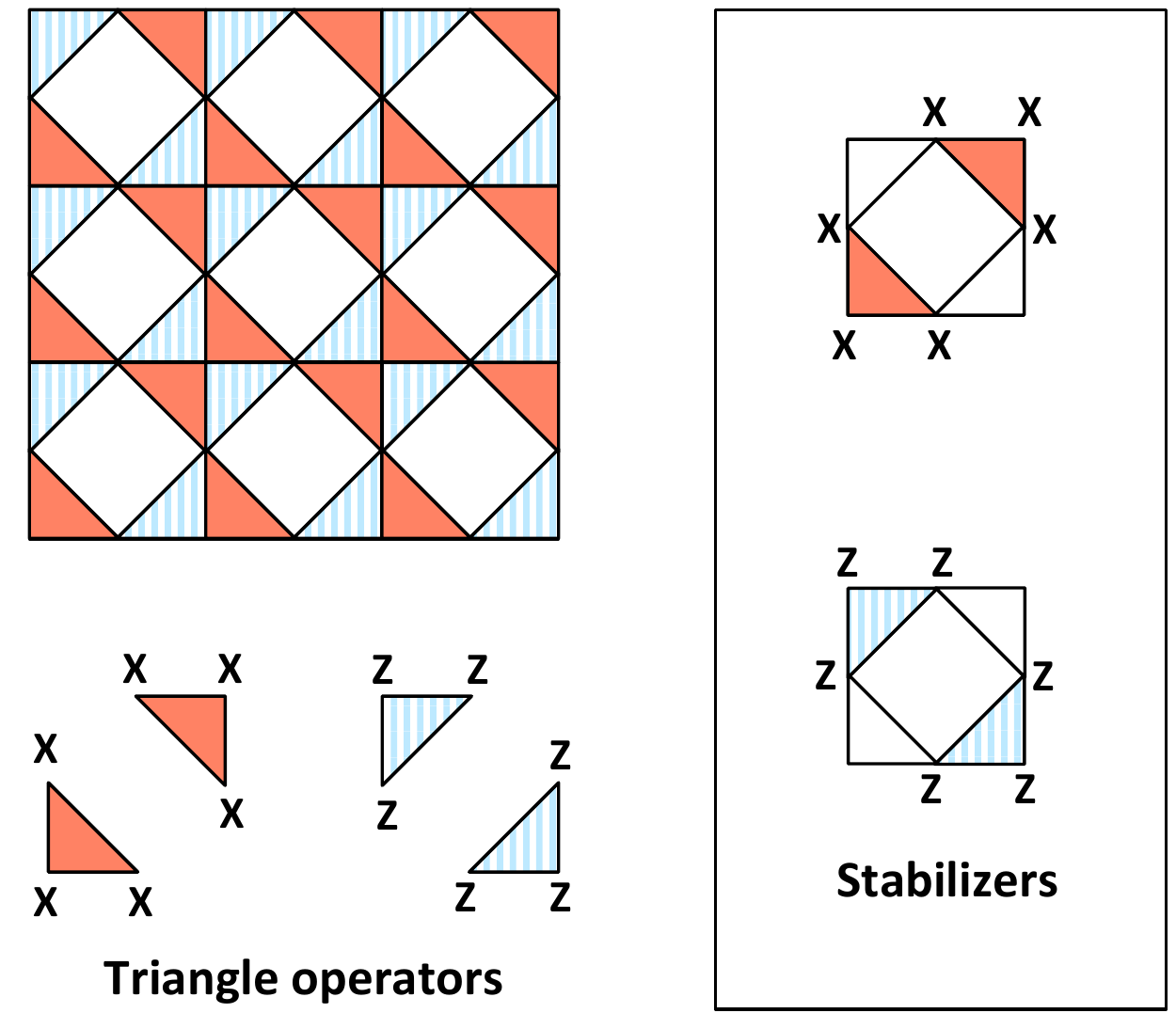}}
\caption{Subsystem  toric code.
Qubits live at vertices and centers of edges of the regular square lattice.
Opposite sides of the lattice are identified.
{\em Left:} Four types of triangles and the corresponding
triangle operators $G(T)$. Triangle operators that belong to
different plaquettes pairwise commute.
 {\em Right:} Stabilizer operators $S^X_p$ (top) and $S^Z_p$ (bottom).
Stabilizers are analogues of the plaquette and star operators of the standard toric code.
Triangle operators commute with stabilizers.
Eigenvalue of any  stabilizer
can be determined by measuring eigenvalues of individual triangle operators.
For a lattice of size $L\times L$ the code has parameters $[[3L^2,2,L]]$.
This should be compared with the standard surface code which has
parameters $[[2L^2,2,L]]$.
}
\label{fig:code_lattice_toric}
\end{figure}
For each plaquette $p$ define weight-$6$ Pauli operators
$S^X_p$ and $S^Z_p$ as shown on Fig.~\ref{fig:code_lattice_toric} (right).
One can easily check that $S^X_p$ and $S^Z_q$ commute with each other
for all $p$ and $q$. Let $\calS$ be the abelian group
generated by $2L^2$ operators  $S^X_p$ and $S^Z_p$.
It defines a quantum stabilizer code with a codespace $\calC$ spanned by
$n$-qubit states $\psi$ invariant under $\calS$, that is, $\psi\in \calC$
iff $S^X_p\cdot \psi =\psi$ and $S^Z_p\cdot \psi =\psi$ for all $p$.
A simple algebra shows that $\prod_p S^X_p=I$ and $\prod_p S^Z_p=I$, where
the product is taken over all plaquettes of the lattice. Furthermore, since
each qubit belongs to exactly two stabilizers $S^X_p$ and two stabilizers $S^Z_p$,
these are the only dependencies among the generators of $\calS$. This shows that
$\calS$ has $s=2(L^2-1)$ independent generators. The standard stabilizer
formalism~\cite{NCbook} then implies that $\calS$ is a stabilizer code encoding
$k'=n-s=L^2+2$ qubits, that is, $\mathrm{dim}(\calC)=2^{k'}$.

We shall now divide the $k'$ encoded qubits
into $g=L^2$ gauge qubits and $k=2$ logical qubits.
Let $u$ be any vertex of the lattice and $f,g$
be a pair of orthogonal edges incident to $u$. The triple $(u,f,g)$
will be referred to as a {\em triangle}.
Note that the lattice has four non-equivalent types of triangles,
see Fig.~\ref{fig:code_lattice_toric}.
We shall say that a triangle $T=(u,f,g)$ is north-west (NW)
if $u$ is at the north-west corner of the plaquette formed by $f$ and $g$.
Similarly one defines north-east (NE), south-west (SW),
and south-east (SE) triangles.
Define {\em triangle operators}
\be
\label{G(T)}
G(T)=\left\{  \ba{rcl} X_u X_f X_g &\mbox{if} & \mbox{$T$ is SW or NE triangle}, \\
Z_u Z_f Z_g&\mbox{if} & \mbox{$T$ is SE or NW triangle}, \\
\ea
\right.
\ee
see Fig.~\ref{fig:code_lattice_toric}. Here subscripts indicate qubits acted upon by
the Pauli operators $X$ and $Z$. Note that triangle operators that belong to different
plaquettes commute with each other.

Any stabilizer can be expressed as a product of two triangle operators
using identities
\bea
S^X_p&=&G(T^{SW}_p) G(T^{NE}_p), \nn \\
S^Z_p&=&G(T^{SE}_p) G(T^{NW}_p),\label{stabilizers}
\eea
see Fig.~\ref{fig:code_lattice_toric}. Here $T^{NW}_p$, $T^{NE}_p$,
$T^{SW}_p$, and $T^{SE}_p$ are triangles of type
NW, NE, SW, SE respectively that belong to a plaquette $p$.

Let us now show that triangle operators commute with all stabilizers, thus
being
logical operators of the code $\calS$.
Consider any stabilizer, say, $S^X_p$ and any triangle operator $G(T)$
of $Z$-type.  If $T$ does not belong to the plaquette $p$ then $S^X_p$ commutes
with $G(T)$ because  triangle operators from different plaquettes always commute.
If $T$ belongs to $p$ then $G(T)$ anti-commutes with both $X$-type triangles
forming $S^X_p$, that is, $G(T)$ commutes with $S^X_p$.
A similar argument shows that $Z$-type stabilizers commute with $X$-type triangle
operators.

The above observations show that we can choose $g=L^2$ pairs of logical
operators for the code $\calS$ as
\be
\label{gauge_qubit}
\overline{X}_p=G(T^{SW}_p) \quad \mbox{and} \quad \overline{Z}_p=G(T^{SE}_p),
\ee
where $p$ runs over all plaquettes of the lattice.
We shall treat encoded qubits defined
by $\overline{X}_p$ and $\overline{Z}_p$ as gauge qubits
encoding no useful information  because the corresponding logical operators
have very small weight.  Hence each plaquette of the lattice carries
one gauge qubit. As we will show in Sec.~\ref{sec:TQO}, it is possible to completely disentangle
these gauge qubits from the code with a depth 4 local quantum circuit, leaving behind the
usual toric code on $2L^2$ qubits and $L^2$ ancillary qubits that are decoupled from the code.

Recall that the code $\calS$ has $k'=L^2+2$ encoded qubits.
This leaves $k=k'-g=2$ logical qubits which have not been identified yet.
 Let $\Gamma$ and $\Lambda$ be the set of all qubits lying on some
fixed horizontal and some fixed vertical line of the lattice respectively, see Fig.~\ref{fig:logicals_toric}.
Note that $|\Gamma|=|\Lambda|=2L$.
 Define
\be
\label{logical1}
\overline{X}_1=\prod_{j\in \Gamma} X_j, \quad \overline{Z}_1=\prod_{j\in \Lambda} Z_j
\ee
and
\be
\label{logical2}
\overline{X}_2=\prod_{j\in \Lambda} X_j, \quad \overline{Z}_2=\prod_{j\in \Gamma} Z_j.
\ee
 \begin{figure}[h]
\centerline{\includegraphics[height=4cm]{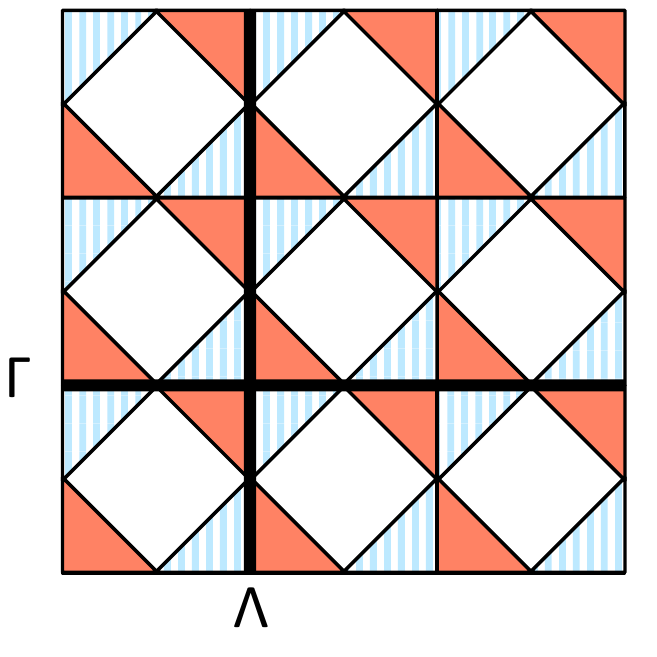}}
\caption{Non-contractible loops on the toric code lattice.
Each loop contains $2L$ qubits.
}
\label{fig:logicals_toric}
\end{figure}

Any triangle operator commutes with $\overline{X}_j$ and $\overline{Z}_j$
because triangles share $0$ or $2$ qubits with $\Gamma$ and $\Lambda$.
In addition, since $\Gamma$ and $\Lambda$ are even sets of qubits
overlapping on exactly one qubit, one has commutation rules
$\overline{X}_i \overline{Z}_j=(-1)^{\delta_{i,j}}
\overline{Z}_j\overline{X}_i$. Hence one can view  $\overline{X}_1,\overline{Z}_1$ and
$\overline{X}_2,\overline{Z}_2$  as logical $X$ and $Z$ operators on the two remaining logical qubits
of the code $\calS$.

The minimum distance $d$ of a subsystem code is defined as the minimum weight
of a Pauli error $E$ that commutes with all stabilizers and implements a non-trivial
transformation on the logical qubits, see~\cite{Poulin05}.
Let us show that the subsystem toric code has distance $d=L$.
Indeed, any error $E$ as above must anti-commute with at least one of the logical operators
$\overline{X}_1,\overline{Z}_1,\overline{X}_2,\overline{Z}_2$.
Assume wlog that $E$ anti-commutes with $\overline{Z}_1$.
Let $\Lambda'$ be any vertical line on the lattice (a horizontal translation of $\Lambda$)
and $\overline{Z}_1'=\prod_{j\in \Lambda'} Z_j$. One can easily check that
$\overline{Z}_1 \overline{Z}_1'$ coincides with the product of stabilizers
$S^Z_p$ over all plaquettes $p$ lying between $\Lambda$ and $\Lambda'$.
Since $E$ commutes with all stabilizers, we conclude that $E$ anti-commutes with
$\overline{Z}_1'$. But this means that  $E$ must act non-trivially on at least one qubit
of $\Lambda'$. Since there are $L$ non-overlapping choices of the line $\Lambda'$, we
conclude that $E$ must have weight at least $L$. One can also easily check
that translating $\Lambda$ by the half of the lattice period gives
a logical operator of weight $L$ equivalent to $\overline{Z}_1$
(modulo gauge operators). Hence the code has distance $d=L$.

The first step in any error correction protocol based on stabilizer codes is
the {\em syndrome readout}, that is, a non-destructive eigenvalue measurement
of every stabilizer operator. To measure the $6$-qubit stabilizers $S^Z_p$ and $S^X_p$
we shall take advantage of the gauge qubits  and the identity Eq.~(\ref{stabilizers}).
The simplest syndrome readout protocol consists of two steps: {\em Step 1.}
Measure the eigenvalue of every $X$-type triangle operator $G(T)$
and record the outcome $\lambda(T)=\pm 1$.
{\em Step 2.} Measure the eigenvalue of every $Z$-type triangle operator $G(T)$
and record the outcome $\lambda(T)=\pm 1$.
Since any triangle operator commutes
with stabilizers, the eigenvalue of any stabilizer remains unchanged throughout the
execution of  this protocol. Hence the eigenvalues of
stabilizers $S^Z_p$ and $S^X_p$ are given by
$\lambda(S^Z_p)=\lambda(T^{SE}_p) \lambda(T^{NW}_p)$
and
$\lambda(S^X_p)=\lambda(T^{SW}_p) \lambda(T^{NE}_p)$,
see Eq.~(\ref{stabilizers}).
In practice it may be advantageous to use `interleaved' protocols in which
Steps~1,2 defined above are implemented in parallel, see Section~\ref{sec:circuit} for more details.

\section{Topological quantum order}
\label{sec:TQO}

Consider a Hamiltonian
\be
\label{H}
H=-\sum_T G(T),
\ee
where the sum is over all triangles of the lattice.
Recall that $G(T)$ are the $3$-qubit triangle operators defined in
Eq.~(\ref{G(T)}).
In this section we compute the entire eigenvalue spectrum of $H$
and show that on the torus $H$ has a four-fold degenerate ground state
separated from excited states by a constant energy gap.
Moreover, we shall construct  a unitary locality preserving
transformation $U$ such that
$UHU^\dag$ can be regarded as the standard toric code Hamiltonian
on the square lattice (with some irrelevant ancillary qubits).
Thus the model defined in Eq.~(\ref{H}) exhibits topological quantum order.

Let us first compute eigenvalues of $H$.
Since the stabilizers $S^X_p$, $S^Z_p$ commute with every term in $H$,
we can assume that any eigenvector $\psi$ of $H$ is also an eigenvector
of any stabilizer, that is, $S^X_p\psi=x_p\psi$ and $S^Z_p\psi=z_p \psi$
for some syndromes $x_p,z_p=\pm 1$.
Using identities Eq.~(\ref{stabilizers},\ref{gauge_qubit}) one gets
\[
H\psi=-\sum_p (1+x_p) \overline{X}_p\psi + (1+z_p) \overline{Z}_p\psi
\]
where $\overline{X}_p,\overline{Z}_p$ are the logical operators on the gauge qubit
located at the plaquette $p$.
Hence the restriction of $H$ onto the sector with fixed syndromes
$x_p,z_p$ describes $L^2$ non-interacting gauge qubits.

Let $\epsilon_0(x_p,z_p)$ and $\epsilon_1(x_p,z_p)$
be the smallest and the largest eigenvalues of a gauge qubit $p$ for a fixed syndromes $x_p,z_p$. A simple algebra shows that
\begin{center}
\begin{tabular}{|c|c|c|c|}
\hline
$x_p$ & $z_p$ & $\epsilon_0(x_p,z_p)$ & $\epsilon_1(x_p,z_p)$  \\
\hline
$1$ & $1$ & $-2\sqrt{2}$ & $2\sqrt{2}$ \\
$1$ & $-1$ & $-2$ & $2$ \\
$-1$ & $1$ & $-2$ & $2$ \\
$-1$ & $-1$ & $0$ & $0$ \\
\hline
\end{tabular}
\end{center}
To minimize the overall energy one has to choose $x_p=z_p=1$ for all $p$.
This shows that ground states of $H$ belong to the trivial syndrome sector
and the ground state energy is   $E_0=-2\sqrt{2}L^2$.
The ground state is four-fold degenerate since the code has
two logical qubits.

Excitations of $H$ fall into two categories. First, there are gauge excitations
that are confined to the trivial syndrome subspace $x_p=z_p=1$.
The energy cost of a single gauge excitation is $\Delta_g=\epsilon_1(1,1)-\epsilon_0(1,1)=4\sqrt{2}$.
A gauge excitation on a plaquette $p$ can be created locally
by a proper combination of operators $\overline{X}_p$ and $\overline{Z}_p$.
Hence gauge excitations do not carry any topological charge.
Secondly, there are syndrome excitations that flip syndrome bits $x_p$ and $z_p$.
The energy cost of a single syndrome excitation
is $\Delta_s=\epsilon_0(1,-1)-\epsilon_0(1,1)=2(\sqrt{2}-1)$. It corresponds to flipping $x_p$ (or $z_p$) on any plaquette $p$.
A single syndrome excitation however cannot be created locally due to the
constraints $\prod_p x_p=\prod_p z_p=1$, see the previous section.
It means that syndrome excitations can only be created in pairs.
Each pair costs energy $2\Delta_s$.

We can now show that the Hamiltonian of Eq.~\eqref{H} is locally equivalent to Kitaev's toric code. Consider a quantum circuit $U$ shown on Fig.~\ref{decoupling_circuit}.
It consists of four rounds of CNOT gates, $U=U^{(4)} U^{(3)}U^{(2)}U^{(1)}$,
where $U^{(j)}$ is a tensor product of $L^2$ disjoint CNOT gates
labeled by $j$ on Fig.~\ref{decoupling_circuit}. Note that $U$ is a locality preserving
transformation, that is, the Heisenberg evolution of any observable
$O\to UOU^\dag$ can only enlarge the support of $O$ by a few units of length.
Such transformations do not change any topological features of the model~\cite{CGW10a}. A simple algebra shows that
the transformed stabilizers $US^X_pU^\dag\equiv A_p$ and $US^Z_p U^\dag\equiv B_p$ coincide with the star and plaquette operators of the Kitaev's toric code
on a tilted square lattice, see Fig.~\ref{decoupling_circuit}.
Furthermore, the transformed gauge generators $U\overline{X}_pU^\dag\equiv J^X_p$
and $U\overline{Z}_pU^\dag\equiv J^Z_p$ become one-qubit Pauli operators
$X$ and $Z$ respectively acting on the qubit located at the bottom edge of $p$,
see Fig.~\ref{decoupling_circuit}.
Using identities Eqs.~(\ref{stabilizers},\ref{gauge_qubit}) we arrive at
\[
H'\equiv UHU^\dag =-\sum_p J^X_p + J^X_p A_p + J^Z_p + J^Z_p B_p.
\]
The same arguments as above show that ground states  of $H'$
are defined by equations $A_p\psi_0=B_p\psi_0=\psi_0$ and
$(J^X_p+J^Z_p)\psi_0=\sqrt{2}\psi_0$ for all plaquettes $p$.
Thus any ground state of $H'$ must have a form
$\psi_0=\psi_{top}\otimes \psi_{anc}$,
where $\psi_{top}$ is a ground state of the Kitaev's toric code
on the tilted square lattice while $\psi_{anc}$ is a tensor product
of one-qubit ancillary states located on horizontal edges of the original
lattice.  Such ancillary unentangled states clearly have no effect on topological features
of the model. We conclude that  that the Hamiltonian Eq.~(\ref{H})
is in the same topological phase as the toric code model.
The exact solvability of the model clearly
extends to a more general Hamiltonian $H=-\sum_T g_T G(T)$, where $g_T$ are arbitrary coefficients.

Consider now a modified Hamiltonian
\[
H''=-\sum_p J^X_p + J^Z_p + A_p +B_p.
\]
Note that $H'$ and $H''$ have the same ground subspace and $H''$ coincides with the ordinary toric code Hamiltonian~\cite{Kitaev03}
if one ignores the ancillary qubits.
For any state $\phi$ orthogonal to the common ground
subspace of $H',H''$ and for any parameter $0\le t\le 1$ one has
\[
\la \phi|(1-t)H' + t H''|\phi\ra \ge (1-t)\Delta' + t\Delta''
\]
where $\Delta'=4(\sqrt{2}-1)$ and $\Delta''=2$ are the energy gaps of $H'$ and $H''$ respectively. It follows that the Hamiltonian $(1-t)H'+tH''$ has energy
gap at least $\min{(\Delta',\Delta'')}=4(\sqrt{2}-1)$ for all $0\le t\le 1$.
 Hence we can continuously deform $H'$ to $H''$ without closing the gap.

The decoupling circuit of Fig.~\ref{decoupling_circuit} illustrates a feature of a more general family of 2D translationally-invariant subsystem codes described in \cite{Bombin11}: they can be mapped to (one or more copies of) the standard toric code with additional decoupled qubits by a constant-depth quantum circuit.

\begin{figure}[h]
\centerline{\includegraphics[height=5.5cm]{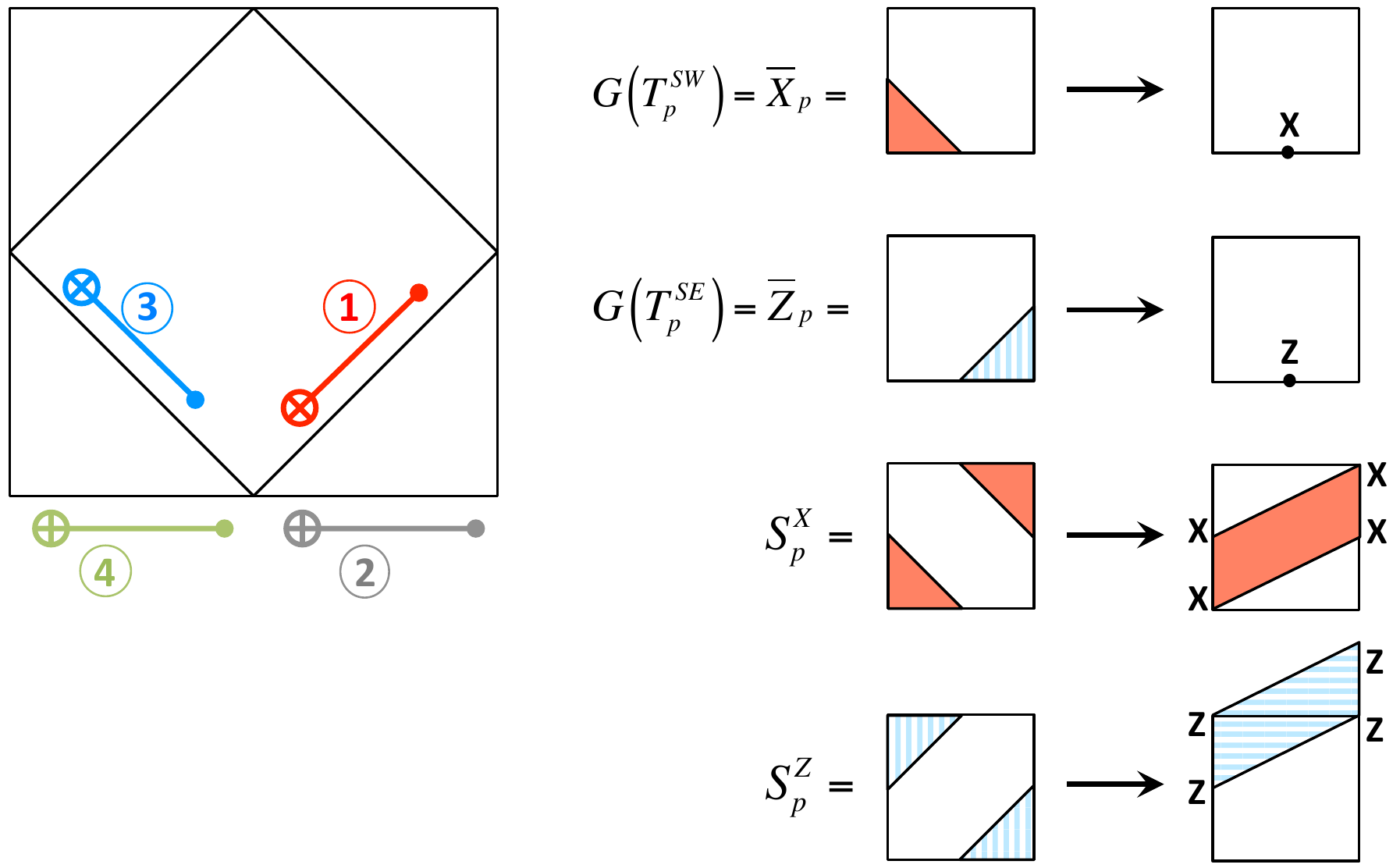}}
\caption{Decoupling circuit. Applying four rounds of CNOT gates  as shown on the left in a translational invariant way transforms the gauge operators and stabilizers as shown on the right. The stabilizer generators become those of Kitaev's toric code
on a tilted square lattice, while the extra gauge generators are mapped to single-qubit Pauli operators
 acting on ancillary qubits (horizontal edges).
 Thus, the circuit has the effect of decoupling the gauge operators from the toric code.}
\label{decoupling_circuit}
\end{figure}
We note that in general the Hamiltonian obtained by the sum of the gauge operators of a topological subsystem code does not necessarily produce topological order.
The peculiar feature of  the present model is that the canonical logical Pauli operators
on the gauge qubits $\overline{X}_p$ and $\overline{Z}_p$ are {\em local}.
In contrast, it was shown in \cite{Bombin11'} that the subsystem color code can be obtained from multiple copies of the ordinary  color code by gauging out both local and non-local logical operators. In particular, some of the logical operators that are gauged out carry topological charge. The present analysis does not apply to such models.


\section{Subsystem surface codes}
\label{sec:SSC}

We can now describe a subsystem version of the simplest
surface code on a planar lattice with two rough and  smooth boundaries that
encodes one logical qubit~\cite{BK:surface}.
Now the lattice has open boundary
conditions. A lattice of size $L\times L$ has $(L+1)^2$ vertices,
$2L(L+1)$ edges, and $L^2$ plaquettes. As before, code qubits
are placed at vertices and centers of edges, so the total number of code qubits
is $n=(L+1)^2+2L(L+1)=3L^2+4L+1$.

For every edge $e$ lying on the boundary of the lattice
define a weight-$2$ Pauli operator  $S_e$ of type $XX$ or $ZZ$ as shown on Fig.~\ref{fig:code_lattice}.
Let $\calS$ be  the group generated by weight-$6$ operators
$S^X_p$, $S^Z_p$ defined earlier and all weight-$2$ operators $S_e$
associated with boundary edges. One can easily check that $\calS$
is an abelian group with $s=2L^2+4L$ independent generators.
Therefore it defines a stabilizer code encoding $k'=n-s=L^2+1$ qubits.

\begin{figure}[h]
\centerline{\includegraphics[height=7cm]{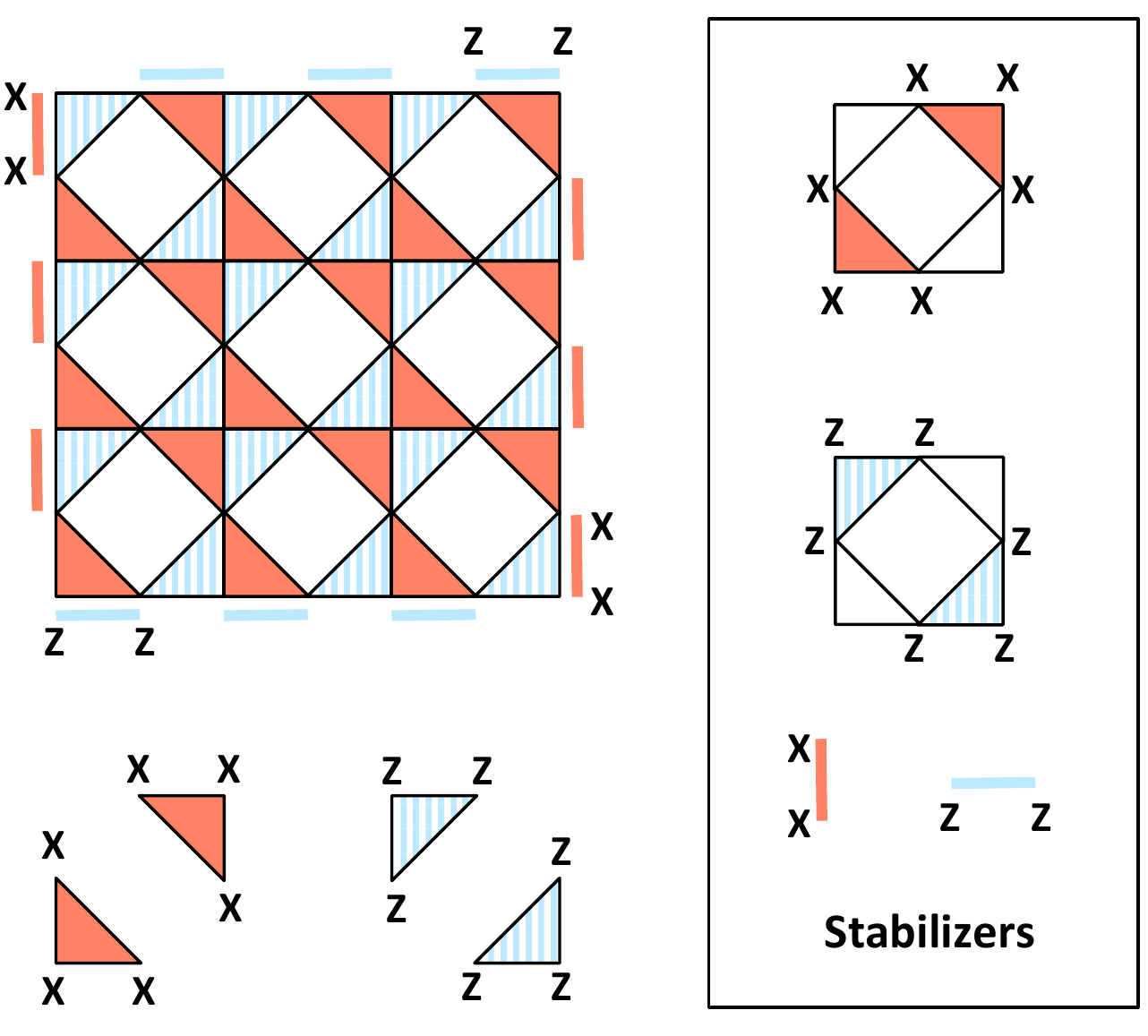}}
\caption{Subsystem version of the surface code.
}
\label{fig:code_lattice}
\end{figure}

Since the additional weight-$2$ stabilizers $S_e$ lying on the boundary
commute with all triangle operators $G(T)$, we can use Eq.~(\ref{gauge_qubit})
to define $g=L^2$ gauge qubits associated with plaquettes of the lattice.
Logical operators on the remaining $k=k'-g=1$ logical qubit can be chosen
as $\overline{X}_1$ and $\overline{Z}_1$, see Eq.~(\ref{logical1}),
that is, a horizontal line of $X$'s and a vertical line of $Z$'s. Note that
a vertical line of $X$'s and a horizontal line of $Z$'s are no longer logical operators
because they anti-commute with some of the boundary stabilizers $S_e$.
The same arguments as above imply that the code has minimum distance $d=L$.
An extension to a planar lattice with punctured holes encoding multiple logical qubits
is sketched on Fig.~\ref{fig:lattice_with_holes}.

\begin{figure}[h]
\centerline{\includegraphics[height=6cm]{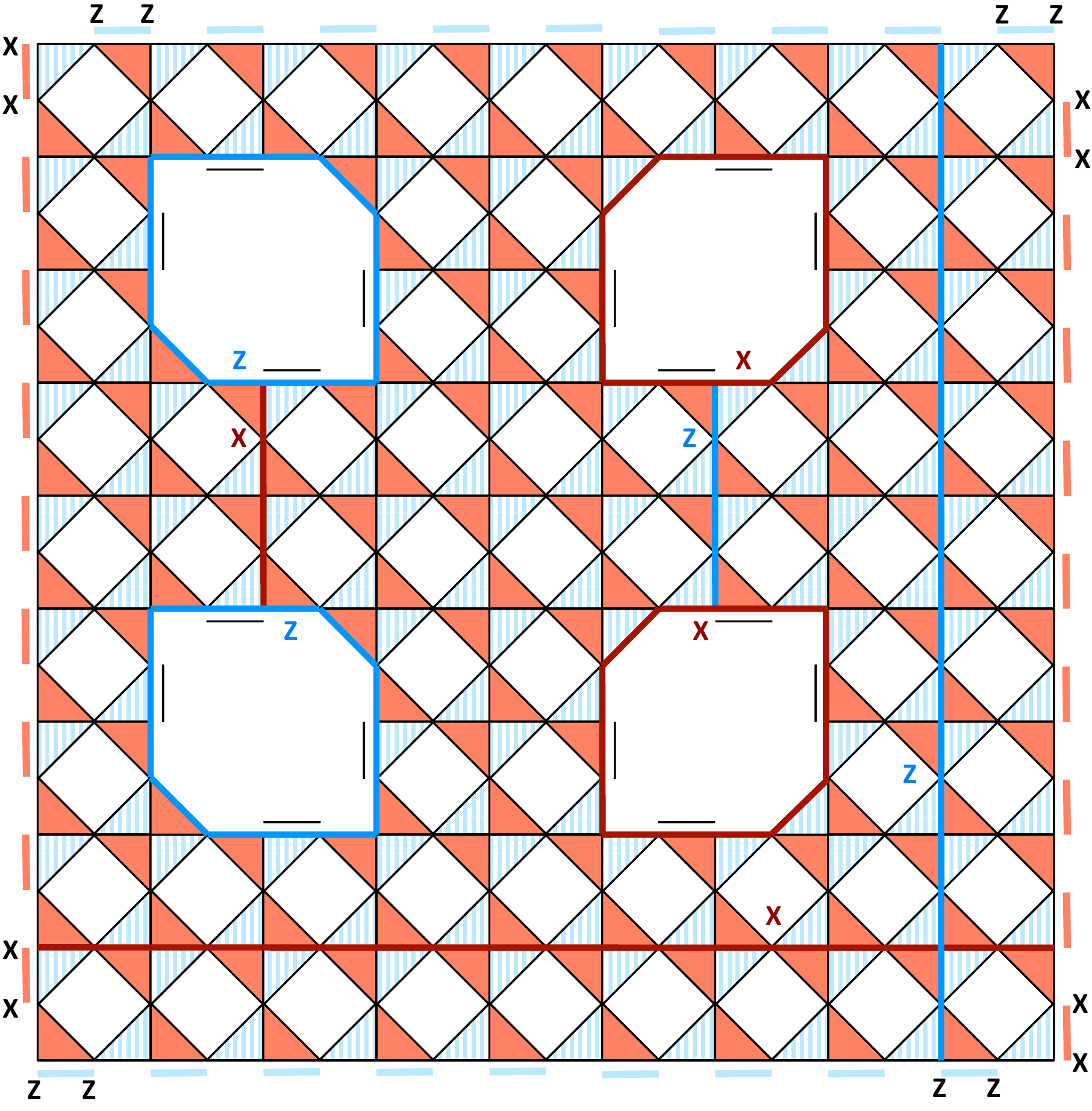}}
\caption{Subsystem version of the surface code with two punctured holes.
Each pair of holes encodes one logical qubit. The external boundary of the lattice
generates the third logical qubit.
Thick red and blue lines show the logical operators.
}
\label{fig:lattice_with_holes}
\end{figure}

\section{Error correction with noiseless syndromes}
\label{sec:ec2D}

In this section we propose an error correction protocol for the idealized setting
where the syndrome readout is noiseless. For simplicity, we shall first focus on
the subsystem toric code defined in Section~\ref{sec:STC}.
It suffices to construct a protocol for correcting errors of $X$-type (bit flips).
Due to the code symmetry, the same protocol can be applied to $Z$-type errors.
 Let $E$ be an unknown Pauli error of $X$-type. We will say that $E$
creates a {\em defect} at a plaquette $p$ iff $E$ anti-commutes with the stabilizer
$S^Z_p$. The syndrome measurement reveals a configuration of defects created
by $E$.  The key observation is that any single-qubit $X$-error
creates exactly two defects.  Indeed, an $X$ error on a vertical or horizontal
edge $e$ creates defects at the two plaquettes adjacent to $e$.
An $X$ error at any vertex $u$ creates defects at the two plaquettes
lying in the north-west and the south-east quadrants of $u$.
The relationship between single-qubit $X$-errors and the corresponding
pairs of defects can be captured by introducing a {\em virtual lattice} $\Lambda$
that consists of virtual vertices and virtual edges.
A virtual vertex $p$ represents a stabilizer $S^Z_p$
(plaquette $p$ of the original lattice), while a virtual edge  represents a
pair of defects that can be created by a single-qubit $X$ error.
One can easily check that
$\Lambda$ is the regular triangular lattice, see Fig.~\ref{fig:virtual_lattice_toric}.

Furthermore, for each virtual edge $e=(p,q)$ there is only one $X$-error
creating a pair of defects at $p$ and $q$.
For any virtual vertex $p$ let $\delta(p)$ be the set of virtual edges incident to $p$.
Then
\[
S^Z_p=\prod_{e\in \delta(p)} Z_e,
\]
that is, $Z$-type stabilizers can be regarded as star operators of the standard
toric code defined on the triangular lattice. Note that
a closed loop on the virtual lattice enclosing any triangular face
is an $X$-type triangle operator $G(T)$, while non-contractible closed
loops correspond to logical operators.

\begin{figure}[h]
\centerline{\includegraphics[height=5cm]{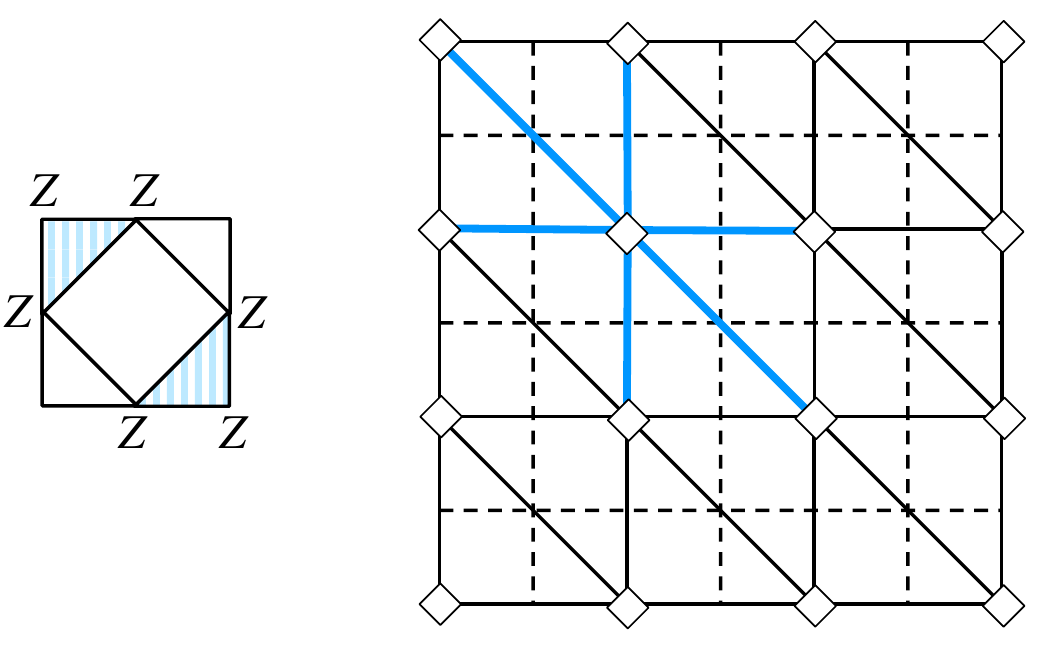}}
\caption{The virtual lattice $\Lambda$ (solid lines) describing correction of $X$-type errors
for the subsystem toric code of Fig.~\ref{fig:code_lattice_toric}.
Opposite sides of the lattice must be identified.
The original square lattice is shown by dashed lines.
Each virtual edge represents a code qubit.
Stabilizers of $Z$-type $S^Z_p$ correspond to stars on the virtual lattice
(solid blue lines).  Triangle operators of $X$-type correspond to triangular faces
of the virtual lattice (not shown).
Error correction amounts to finding the minimum weight matching of
defects on the virtual lattice.
The virtual lattice describing correction of $Z$-type errors is obtained from $\Lambda$ by the
 $90^\circ$ rotation.}
\label{fig:virtual_lattice_toric}
\end{figure}

Assuming that errors on different qubits are independent and have the same rate $p$,
the most likely error $E^*$ consistent
with the observed syndrome coincides with the minimum weight
matching of defects on the virtual lattice.
The latter can be found efficiently using the Edmonds's algorithm,
see~\cite{Dennis01} for details.
Choosing $E^*$ as a correction operator always returns the system back to the
codespace $\calC$. The overall evolution of the system is described by
an operator $EE^*$ which has trivial syndrome and thus can be viewed as a
linear combination of closed loops on the virtual lattice with $\ZZ_2$ coefficients.
Error correction is successful iff $EE^*$ acts non-trivially only on the gauge qubits,
that is, $EE^*$ is a product of $X$-type triangle operators $G(T)$.
Equivalently, $EE^*$ must represent the trivial cycle in the homology group
$H_1(\Lambda,\ZZ_2)$.

As was argued in~\cite{Dennis01,Poulin09,bombin:topsub},
the optimal error correction strategy amounts to
finding the most likely {\em equivalence class} of errors
consistent with the observed syndrome rather than the most likely error.
More specifically, let $\calG$ be the {\em gauge group} generated by the triangle operators $G(T)$
of $X$-type. Since logical operators acting on gauge qubits are irrelevant,
all errors in the coset $E\cdot \calG$ must be considered equivalent.
Note that there are only four cosets of $\calG$
consistent with the observed syndrome, namely, $E\cdot \calG$, $E\overline{X}_1 \cdot \calG$,
$E\overline{X}_2 \cdot \calG$, and $E\overline{X}_1 \overline{X}_2 \cdot \calG$.
As was shown in~\cite{Dennis01,bombin:topsub}, the
probability of each coset can be expressed as the partition function
of the random bond $\pm J$ Ising model. In our case Ising spins
reside on triangular faces of the virtual lattice, anti-ferromagnetic bonds
correspond to virtual edges that belong to $E$, and the inverse temperature $\beta$
is determined by the Nishimori condition
$e^{-2\beta J}=\frac{p}{1-p}$.
The threshold error rate $p_c$  for the optimal decoding
coincides with the density of anti-ferromagnetic bonds
at the phase transition point~\cite{Dennis01}.
The latter has been recently computed by Queiroz~\cite{Queiroz06}
who found $p_c\approx 7\%$. The analogous threshold error rate for the standard toric code is known to be approximately $11\%$, see~\cite{Dennis01}.

One can similarly construct the virtual lattice for the subsystem surface code,
see Fig.~\ref{fig:virtual_lattice_open}. The only difference is that now
defects can be matched either to each other, or to one of the boundaries.

\begin{figure}[h]
\centerline{\includegraphics[height=5cm]{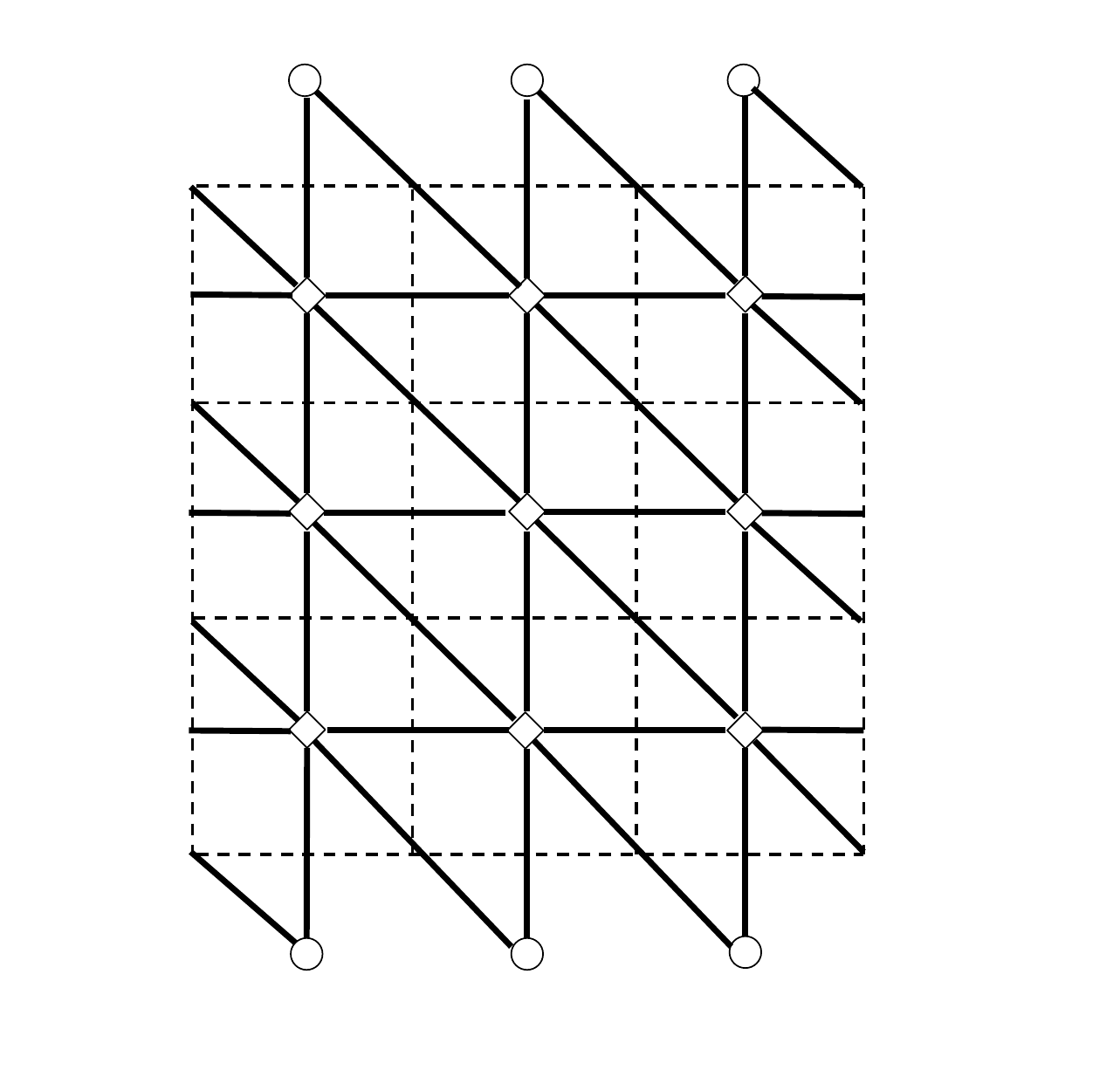}}
\caption{The virtual lattice $\Lambda$ (solid lines) describing correction of $X$-type errors
for the subsystem surface code of Fig.~\ref{fig:code_lattice}.
Diamonds and open circles represent stabilizers
$Z$-type stabilizers
$S^Z_p$ and $S_e$ respectively.
The virtual lattice describing correction of $Z$-type errors is obtained from $\Lambda$ by the
 $90^\circ$ rotation.}
\label{fig:virtual_lattice_open}
\end{figure}

\section{Error correction for the circuit-based error model}
\label{sec:circuit}

Let us now consider more realistic settings when the syndrome information
itself may contain errors.
We assume that the library of elementary operations supported by the quantum hardware includes
CNOT gates between nearest-neighbor qubits, single-qubit measurements in $X$- or $Z$-basis, and
preparation of single-qubit ancillary states $|0\ra$ and $|+\ra$.
Our error correction protocol will be defined as a sequence of {\em rounds},
where at each round any qubit can participate in one elementary operation.
We assume that each elementary operation is noisy, so it can fail
with a probability $p$ that we call an {\em error rate}. More precisely,
our error model, borrowed from~\cite{Fowler08},  is defined as follows.
\begin{itemize}
\item A noisy $X$ or $Z$ measurement is the ideal measurement in which the outcome is flipped
with probability $p$.
\item A noisy $|0\ra$ or $|+\ra$ ancilla preparation returns the correct state with probability $1-p$ and the
 orthogonal state with probability $p$.
\item A noisy CNOT gate  is the ideal CNOT followed by
an error $(1-p) Id + p \calD$,
where $Id$ is the identity map and $\calD$ is the fully depolarizing two-qubit map
applying one of $16$ two-qubit Pauli operators with probability $1/16$ each.
\end{itemize}
We assume that
once a qubit has been measured, its state is unknown. To use such a qubit again, it must be
explicitly initialized using the noisy preparation defined above.
We do not need to define memory errors because no qubit will be idle at any round of
our protocol.

In order to measure eigenvalue of individual triangle operators $Z_i Z_j Z_k$ and $X_p X_q X_r$
we shall use quantum circuits shown on Fig.~\ref{fig:ancilla}.
Measuring a single triangle operator requires one ancillary qubit
and five rounds. Similar circuits with one extra CNOT gate were used
in fault-tolerant protocols based on the standard surface code~\cite{Dennis01,Fowler08},
where one has to measure four-qubit plaquette and star operators
$Z^{\otimes 4}$ and $X^{\otimes 4}$.

\begin{figure}[h]
\centerline{\includegraphics[height=6cm]{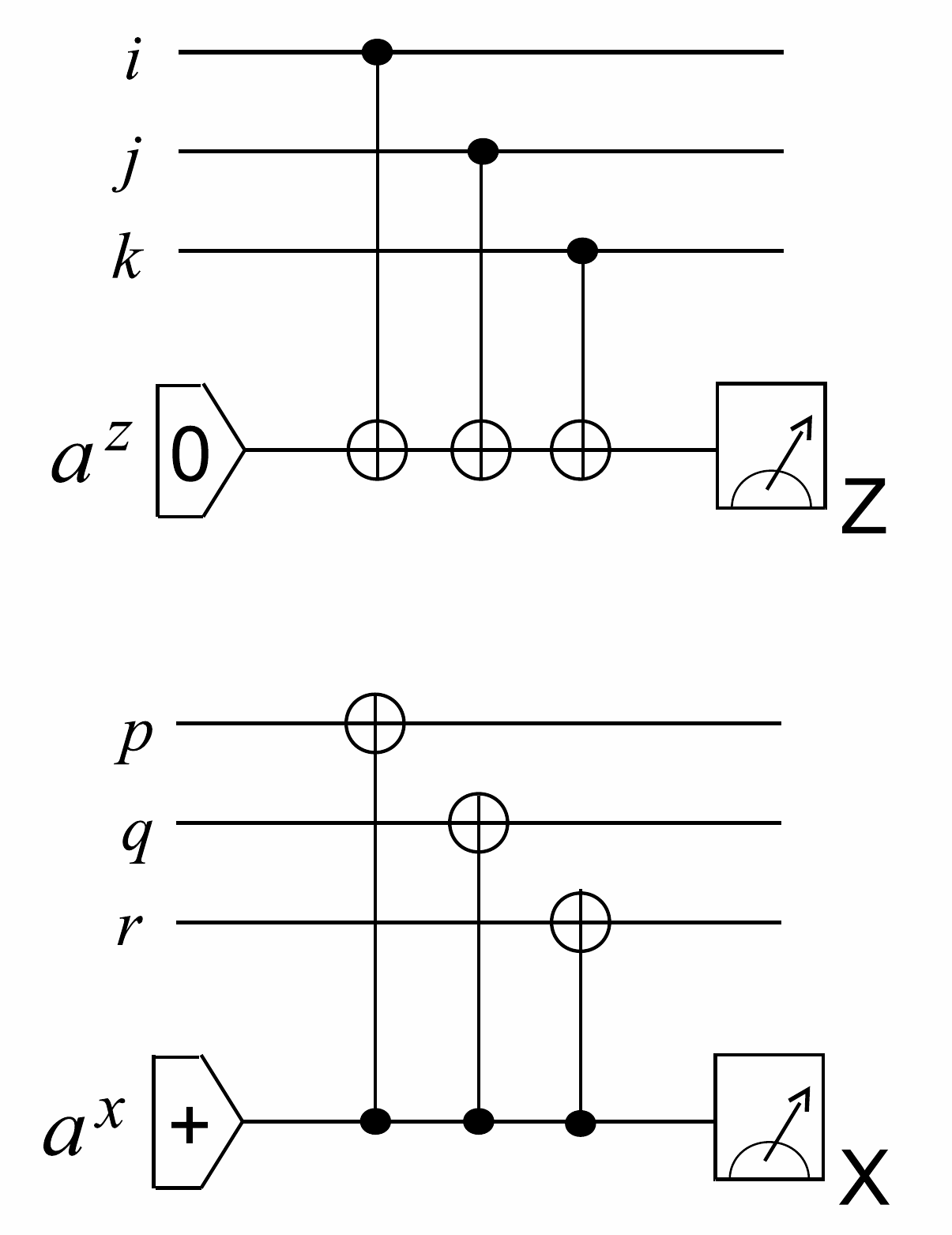}}
\caption{Quantum circuits for
measuring the eigenvalue of triangle operators $Z_i Z_j Z_k$ (top)
and $X_p X_q X_r$ (bottom). The circuits use one ancillary qubit.
A single $Z$ ($X$) error on the ancilla  $a^z$ ($a^x$) can propagate via CNOTs to at most one
$Z$ ($X$) error on code qubits modulo gauge operators.
A single $X$ ($Z$) error on the ancilla  $a^z$ ($a^x$)  results in a faulty
measurement outcome.
}
\label{fig:ancilla}
\end{figure}

We begin by highlighting strengths and weaknesses of the subsystem and the standard
surface code.  The key advantage of the SSC is a limited propagation of errors by the CNOT gates.
Consider, for example, a single $Z$ error
on the ancillary qubit $a^z$ in the circuit measuring the triangle operator $Z_i Z_j Z_k$,
see Fig.~\ref{fig:ancilla} (top). Depending on the round at which the error has occurred, it
propagates to one of the errors $Z_i$, $Z_i Z_j$, $Z_i Z_j Z_k$ on the code
qubits. Multiplying the last two errors by the triangle operator $Z_i Z_j Z_k$
leaves only single-qubit errors $Z_i$, $Z_k$, and the identity error.
It shows that a single $Z$ error on the ancilla $a^z$ can lead to at most one $Z$ error on the
code qubits, modulo gauge operators.  An $X$ error on $a^z$ cannot propagate to code qubits,
so its only effect is
flipping the measurement outcome which results in a faulty syndrome bit.
Since each stabilizer is composed of two disjoint triangle operators,
see Fig.~\ref{fig:code_lattice}, each syndrome bit effectively accumulates errors
from ten rounds. For comparison, the standard surface code
requires six rounds to measure a single syndrome bit, however a single
error on the ancilla can feed back to two errors on code qubits
(such double errors were referred to as `horizontal hooks' in~\cite{Dennis01}).
This shows that neither of the two codes offers an obvious advantage compared
with the other.

Let us now discuss our syndrome readout circuit in more details.
Since individual syndromes can no longer be trusted, we shall repeat
syndrome measurements $T$ times for some $T\gg 1$.
We choose $T=L$ in all numerical simulations.
Error correction is deemed successful if the accumulated error $E$
on the code qubits that results from $T$ noisy syndrome measurements
can be corrected based on the full observed syndrome history
and one final syndrome readout which we assume to be {\em noiseless}.
(In practice the final readout  involves
measuring each code qubit in $|0\ra$ or $|+\ra$ basis.
Outcomes of such measurement determine the syndrome of $Z$-type
or $X$-type stabilizers respectively. We can assume that
single-qubit measurements
are noiseless by  absorbing measurement errors into memory
errors that occurred one round earlier.)

Repeating the circuits shown on Fig.~\ref{fig:ancilla} cyclically $L$ times
would naively require $5L$ rounds. We can reduce the required number
of rounds to $4L$ by  introducing two ancillary qubits for each triangle.
One of them serves as the ancilla $a^z$ or $a^x$ shown
on Fig.~\ref{fig:ancilla}.   The purpose of the second ancilla is
to enable offline  preparation of $|0\ra$ or $|+\ra$ states which can be
performed in the same round as the measurement of the main ancilla.
To simplify notations, we will only show one ancilla per triangle
and assume that this ancilla is initialized in the $|0\ra$ or $|+\ra$
state at the end of each measurement round (with a duly added noise).

The readout circuit will be chosen such that any fixed triangle alternates between
three gate rounds and one measurement round in a cyclic fashion.
It can be represented by a local readout schedule
\be
\label{local_schedule}
\cdots M G_1 G_2 G_3 M G_1 G_2 G_3 M \cdots
\ee
where $M$ is either $X$-type or $Z$-type measurement on the ancilla
while $G_1, G_2, G_3$ are CNOT gates coupling the ancilla and the code qubits
(the ancilla is control for $X$-type triangles and target for $Z$-type triangles).
Time flows from the left to the right. Note that each triangle has $24$ different
choices of its local schedule. Indeed, there are $6$ choices of the order in which the
ancilla is coupled to the code qubits and  $4$ choices of the round at which the triangle is measured.
Local schedules chosen at different triangles must be consistent with each other,
such that  at any round any qubit participates in at most one operation.

We shall focus on schedules which are periodic both in space and time.
Hence the entire readout circuit is completely specified by local schedules
inside a 3D elementary cell which consists of four rounds labeled as  $0,1,2,3$
and four triangles of type NE, SE, SW, NW located at some fixed plaquette $p$.
We begin by observing that a consistent schedule cannot have a round at which
every triangle applies a CNOT. Indeed, this would define a matching between
code qubits and triangles. However, the lattice has $4L^2$ triangles and only
$3L^2$ code qubits. This observation shows that at every round all triangles of some
type have to be measured. Furthermore, it is natural to demand that
if some pair of triangles form a stabilizer, these triangles
must be measured in two consecutive rounds
(otherwise the corresponding syndrome bit would accumulate too much errors).
In other words, we would like to measure all triangles of $X$-type in two consecutive
rounds and all triangles of $Z$-type in the two remaining consecutive rounds.
Our choice of the measurement rounds satisfying these conditions is shown on Fig.~\ref{fig:schedule_meas}.
Note that all $X$-type triangles are measured at rounds $3,0$, while
all $Z$-type triangles are measured at rounds $1,2$.

\begin{figure}[h]
\centerline{\includegraphics[height=4cm]{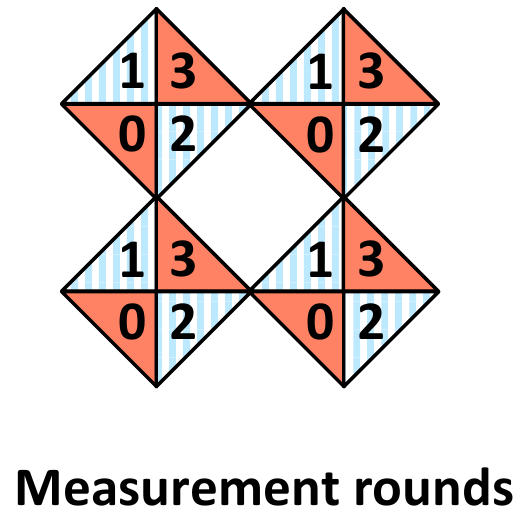}}
\caption{The numbers indicate rounds (modulo four) at which the ancillas assigned to each triangle
have to be measured. }
\label{fig:schedule_meas}
\end{figure}

It remains to schedule CNOT gates.
We will say that a schedule is
{\em correct} iff for each triangle one can move all gates forward in time
towards the next measurement. Here moving a gate is allowed as long as it commutes
with other gates.  A correct schedule faithfully simulates the simple syndrome
extraction routine described in Section~\ref{sec:STC}
since after moving all gates towards the next measurement
all $X$-type stabilizers are measured at rounds $3,0$,
while all $Z$-type stabilizers are measured at rounds $1,2$.
We shall look for a schedule which is correct and
invariant  under the exchange of $X$ and $Z$ (modulo lattice symmetries
and time translations). The latter requirement is rather natural since our
error model does not have a bias towards $X$ or $Z$ errors.

To derive sufficient conditions for correctness let us introduce some terminology.
Consider any triangle $T$ and its local schedule, see Eq.~(\ref{local_schedule}).
We shall refer to the gates $G_1$ and $G_3$ as the {\em first gate} and
the {\em last gate} for the chosen triangle.
If some pair of triangles $T$ and $T'$ are measured  at two subsequent rounds $j$ and $j+1$
respectively,   we will say that $T$ is the {\em leading triangle} and $T'$ is the {\em tailing triangle}.
If $T$ and $T'$ are measured two rounds apart, the choice of the leading and the tailing triangle is
arbitrary.
\begin{lemma}
\label{lemma:correctness}
A schedule of CNOTs  is correct if
for any $X$-type and any $Z$-type triangle at least one of the following
is true:
\begin{itemize}
\item The two triangles are disjoint,
\item The two triangles are measured two rounds apart,
\item The last gate of the leading triangle commutes with the first gate of the tailing triangle.
\end{itemize}
\end{lemma}
\begin{proof}
Suppose $T^x$ and $T^z$ are $X$-type and $Z$-type triangles respectively.
If $T^x$ and $T^z$ are measured two rounds apart,
their combined local schedules can be represented by a diagram
\begin{center}
\begin{tabular}{|c|c|c|c|c|c|c|c|c|}
\hline
$\cdots$ & $M^x$ & $G^x_1$ & $G^x_2$ & $G^x_3$ & $M^x$ & $G^x_1$ & $G^x_2$ & $\cdots$ \\
\hline
$\cdots$ & $G^z_2$ & $G^z_3$ & $M^z$ & $G^z_1$ & $G^z_2$ & $G^z_3$ & $M^z$ & $\cdots$ \\
\hline
\end{tabular}
\end{center}
The gates $G^x_1$ and $G^z_3$ must be disjoint. Similarly, the gates $G^x_3$ and $G^z_1$ must be disjoint.
Hence we can deform the diagram by moving $G^x_1$, $G^z_1$ one round forward and moving
$G^x_3$, $G^z_3$ one round backward obtaining an equivalent circuit:
\begin{center}
\begin{tabular}{|c|c|c|c|c|c|c|c|c|}
\hline
$\cdots$ & $M^x$ & & $G^x_1 G^x_2 G^x_3$ & & $M^x$ &  & $G^x_1 G^x_2 G^x_3$ & $\cdots$ \\
\hline
$\cdots$ & $G^z_1 G^z_2 G^z_3$ &  & $M^z$ &  & $G^z_1 G^z_2 G^z_3$ & & $M^z$ & $\cdots$ \\
\hline
\end{tabular}
\end{center}
We can further deform the circuit by moving each measurement
backwards towards the next gate.

Suppose now that  $T^x$ and $T^z$ are measured in two subsequent rounds such that $T^x$ is the leading and
$T^z$ is the tailing (the opposite case is completely analogous).
Then their combined schedules can be represented by a diagram
\begin{center}
\begin{tabular}{|c|c|c|c|c|c|c|c|}
\hline
$\cdots$ & $M^x$ & $G^x_1$ & $G^x_2$ & $G^x_3$ & $M^x$ & $G^x_1$ & $\cdots$  \\
\hline
$\cdots$ & $G^z_3$ & $M^z$ & $G^z_1$ & $G^z_2$ & $G^z_3$ & $M^z$  &$\cdots$ \\
\hline
\end{tabular}
\end{center}
By assumption, the gates $G^x_3$ and $G^z_1$ commute.
Hence we can deform the diagram by moving $G^x_3$ one round backward
and moving $G^z_1$ one round forward.
In addition, we can move $G^x_1$ one round forward and move $G^z_3$ one round backward.
It gives an equivalent circuit:
\begin{center}
\begin{tabular}{|c|c|c|c|c|c|c|c|c|}
\hline
$\cdots$ & $M^x$ & & $G^x_1 G^x_2 G^x_3 $ & & $M^x$ & & $G^x_1G^x_2 G^x_3 $ & $\cdots$ \\
\hline
$\cdots$ &  & $M^z$ & & $G^z_1G^z_2 G^z_3$ & & $M^z$ &    & $\cdots$ \\
\hline
\end{tabular}
\end{center}
We can further deform the circuit by moving each measurement
backwards towards the next gate.

After the deformation each triangle applies the entire sequence $G_1G_2 G_3 M$
in a single round which is two rounds apart from the measurement round in the original schedule.
We can now move the  sequence $G_1G_2 G_3 M$ two rounds forward for each triangle
simultaneously until $X$-type triangles apply $G^x_1G^x_2 G^x_3 M$ at rounds $0,3$
and  $Z$-type triangles apply $G^z_1G^z_2 G^z_3 M$ at rounds $1,2$.
This shows that the original schedule is correct.
\end{proof}

\begin{figure}[h]
\centerline{\includegraphics[height=5cm]{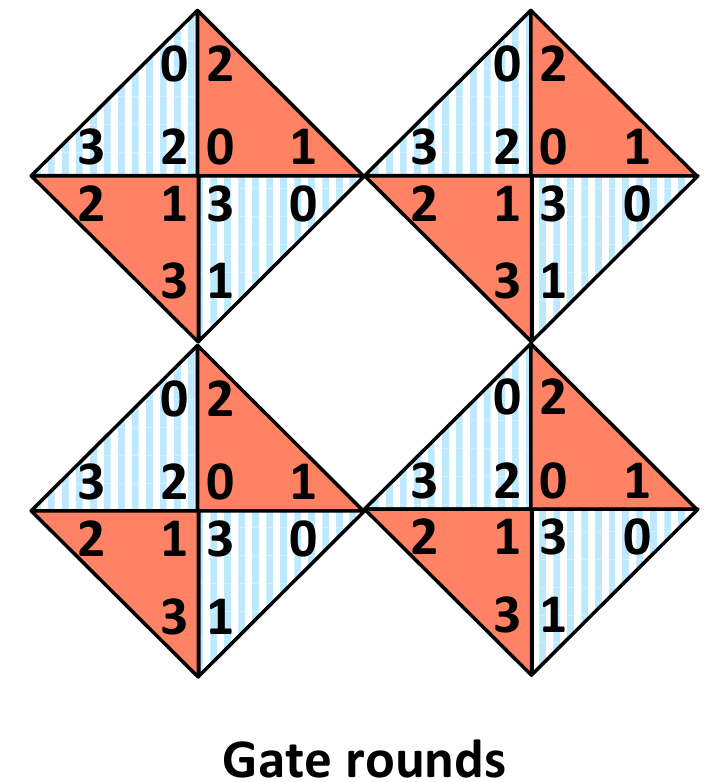}}
\caption{Example of a correct syndrome extraction schedule with four rounds labeled by $0,1,2,3$
repeated cyclically.
The  numbers assigned to vertices of each triangle indicate the rounds at which the code qubits comprising
 a triangle
are coupled to the ancilla by CNOT gates.
Measurement rounds are shown on Fig.~\ref{fig:schedule_meas}.
  One can check that any pair of $x$-type and $z$-type
triangles satisfies conditions of Lemma~\ref{lemma:correctness}. }
\label{fig:schedule_gates}
\end{figure}

Our choice of a CNOT schedule is shown on Fig.~\ref{fig:schedule_gates}.
One can easily check that it satisfies conditions of Lemma~\ref{lemma:correctness}.
It remains to define the classical post-processing step that
extracts the syndromes from the measured eigenvalues of triangle operators.
For any integer $t\in [0,L-1]$ and a plaquette $p$ we define
a syndrome bit $s^Z_p(t)$  as a product of eigenvalues of $Z$-type triangles
SE and NW located at the plaquette $p$  that were measured at rounds $4t+1$ and $4t+2$.
Similarly, we define a syndrome bit $s^X_p(t)$ as a product of eigenvalues of $X$-type triangles
SW and NE  located at the plaquette $p$  that were measured at rounds $4t+3$ and $4t+4$.
Hence each syndrome bit combines eigenvalues of two triangle operators
measured in two consecutive rounds.

Let us now move to the error correction protocol that takes as input the syndrome information
and outputs a correcting Pauli operator $E^*$ acting on the code qubits.
It mostly follows~\cite{Dennis01,Fowler08,Wang11}.
Our protocol deals with $X$-type and $Z$-type errors independently.
It should be noted  that the schedule  shown on Figs.~\ref{fig:schedule_meas},\ref{fig:schedule_gates}
is invariant under the horizontal reflection of the lattice
and shifting the time  by two rounds. Since the horizontal reflection exchanges $X$-type and $Z$-type
triangles, it suffices to analyze $X$-type errors.

Let us introduce a 3D virtual lattice $\Lambda$
that consists of virtual vertices
and virtual edges. A virtual vertex is a pair $(p,t)$, where
$p$ is a plaquette of the 2D code lattice and $t\in [0,L-1]$ is the discrete time.
We shall say that a virtual vertex $u=(p,t)$ has a {\em defect} iff
the syndrome bits $s^Z_p(t)$ and $s^Z_p(t+1)$ are different.
Hence the full syndrome history can be regarded as a configuration of defects
on the virtual lattice.

We begin by considering configurations of defects created by a single fault in the
readout circuit. Here a single fault includes one of the following possibilities:
\begin{itemize}
\item Wrong measurement outcome on some ancilla,
\item Wrong ancilla preparation,
\item One of the errors $IX$, $XI$, or $XX$ inserted after some CNOT gate.
\end{itemize}
In other words, a single fault is any event that can occur with probability $\Omega(p)$
in the limit $p\to 0$ (recall that we only keep track of $X$-type errors).
In order to define virtual edges we need the following observation.
\begin{lemma}
Any single fault in the readout circuit creates either $0$ or $2$ defects on the virtual lattice.
\end{lemma}
\begin{proof}
A measurement error on the ancilla $a^z$ at a plaquette $p$ creates two defects at virtual vertices
$(p,t)$ and $(p,t+1)$ for some $t$. Measurement errors on ancillas $a^x$ create no defects since we
ignore $X$-type syndromes. Ancilla preparation error can be regarded as an $X$-type
error for ancillas $a^z$ and $Z$-type error for ancillas $a^x$. Such an error can be propagated
forward in time without feeding back to the code qubits because $a^z$ is always a target
qubit and $a^x$ is always a control qubit for any CNOT gate, see Fig.~\ref{fig:ancilla}.
 Hence ancilla preparation errors are equivalent to
the measurement errors. By the same reason, an $X$-error on $a^z$ caused by any CNOT
gate is equivalent to a measurement  error.

Consider now a single-qubit $X$ error on some code qubit.
If the syndrome were measured on all plaquettes directly after the error,
one would observe non-trivial syndromes $s^Z_p$, $s^Z_q$
at some pair of plaquettes $p,q$, see Section~\ref{sec:STC}.
Since there are no other errors in the readout circuit, it will faithfully simulate
the ideal syndrome measurements, that is, the syndrome $s^Z_p(t)$
will change from $1$ to $-1$ for some step $t_p$
and the  syndrome $s^Z_q(t)$ will change from $1$ to $-1$ for some step $t_q$.
This produces a pair of defects at virtual sites $(p,t_p)$ and $(q,t_q)$.
(More detailed analysis shows that either $t_p=t_q$ or $t_p=t_q\pm 1$).
An error $XX$ that occurred after a CNOT gate
is equivalent to a single $X$ error on the control qubit that occurred before this CNOT.
The  only remaining case is a single $X$ error on the ancilla $a^x$.
It can be propagated forward or backward towards the nearest
$a^x$-measurement. Such propagation feeds back at most one $X$ error
to code qubits. This is the case that we have already explored.
 \end{proof}
A trivial corollary of the lemma is that the number of defects on the virtual lattice
is always even.

We connect a pair of virtual vertices $u,v$ by a virtual edge
iff there a single fault in the readout circuit capable of creating a pair of defects at $u$ and $v$.
A more detailed analysis shows that the virtual lattice has seven types of edges
(not counting the orientation). The table below shows all $14$ neighbors $v$ of
some fixed virtual vertex $u=(x,y,t)$ as well as the total number of single faults
of each type in the readout circuit that create defects at $u$ and $v$.
For brevity we refer to single faults $IX$, $XI$, $XX$ as G-faults (gate faults),
while measurement and preparation single faults are referred to as
M-faults and P-faults.
\begin{center}
\begin{tabular}{|c|c|c|c|c|}
\hline
Neighbor of $(x,y,t)$ & G-faults & M-faults & P-faults & Prior \\
\hline
$(x,y,t\pm 1)$ & 6 & 2 & 2 & $11p/2$ \\
$(x\pm 1,y,t)$ & 8 & 0 & 0 &  $2p$ \\
$(x\pm 1, y\mp 1,t)$ & 8 & 0 & 0 & $2p$ \\
$(x,y\pm 1,t)$ & 4 & 0 & 0 & $p$ \\
$(x,y\pm 1, t\pm 1)$ & 2 & 0 & 0 & $p/2$\\
$(x\mp 1, y,t\pm 1)$ & 2 & 0 & 0 & $p/2$ \\
$(x\mp 1,y\pm 1,t\pm 1)$ & 2 & 0 & 0 & $p/2$\\
\hline
\end{tabular}
\end{center}
Note that the space-like virtual edges (those for which $u$ and $v$ have the same $t$
coordinate) correspond to edges of the 2D virtual lattice defined in Section~\ref{sec:ec2D}.
If one ignores the $t$ coordinate, space-like virtual edges correspond to the code qubits.
In particular, pairs of defects located on space-like virtual edges can be viewed as
memory errors. On the other hand, pairs of defects located on time-like virtual edges(those for which $u$ and $v$
have the same $x$,$y$ coordinates) can be viewed as syndrome measurement errors.
The remaining virtual edges represent various combinations of memory errors and measurement errors.

For every virtual edge $e$ we define a {\em prior}  $p_e$
as the probability to observe a pair of defects at the endpoints of $e$.
Taking into account that any single G-fault has probability $p/4$, while
a single M-fault and a single P-fault have probability $p$, one arrives at the
priors listed in the table.

We shall choose the correction operator $E^*$ by pretending that the
creation of defect pairs on different virtual edges are independent events.
Then the most likely combination of memory errors and measurement errors
consistent with the observed configuration of defects coincides with the
minimum weight matching of defects on the virtual lattice, where an edge $e$ is assigned a
weight $w_e\sim \log{(1/p_e)}$. The minimum weight matching $M$
can be found efficiently using the Edmonds's algorithm. Finally, we choose
the correction operator $E^*$ as the product of all memory errors that appear in $M$.
In order to decide whether the error correction is successful we compare $E^*$
with the accumulated error $E$ on the code qubits generated by the syndrome readout
circuit. The results of our Monte Carlo simulation are shown on Figs.~\ref{fig:threshold},\ref{fig:overhead}.
It indicates that the threshold error rate for the circuit-based error model is $p_c\approx 0.6\%$.

\begin{figure}[h]
\centerline{\includegraphics[height=6cm]{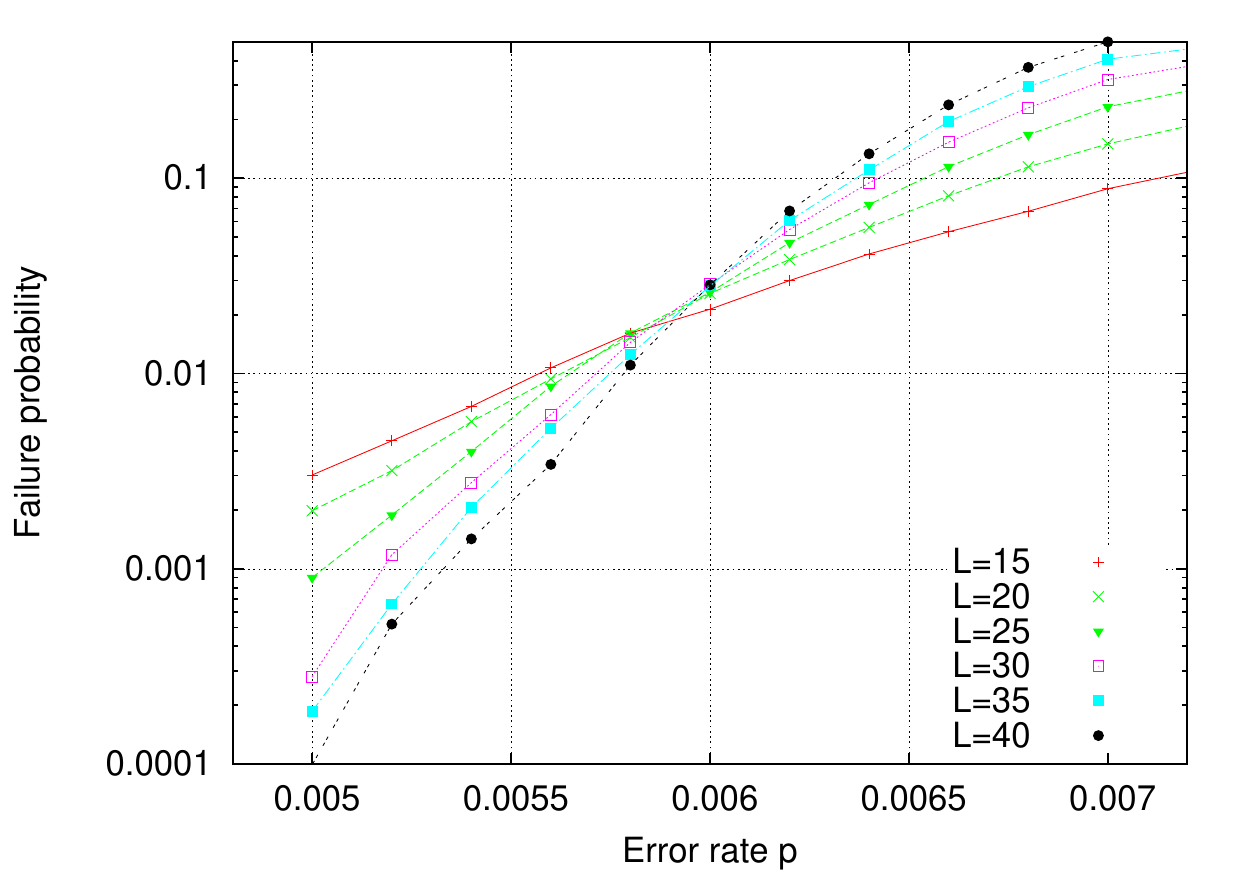}}
\caption{Probability of the error correction failure for the subsystem toric code
under the circuit-based error model.
For a lattice of size $L$  the syndrome measurement has been repeated $L$ times.
Each data point was obtained using $10^4-10^6$ Monte Carlo trials.
The simulation was performed only for $X$-type errors ($Z$-type errors
on the reflected lattice).}
\label{fig:threshold}
\end{figure}

\begin{figure}[h]
\centerline{\includegraphics[height=6cm]{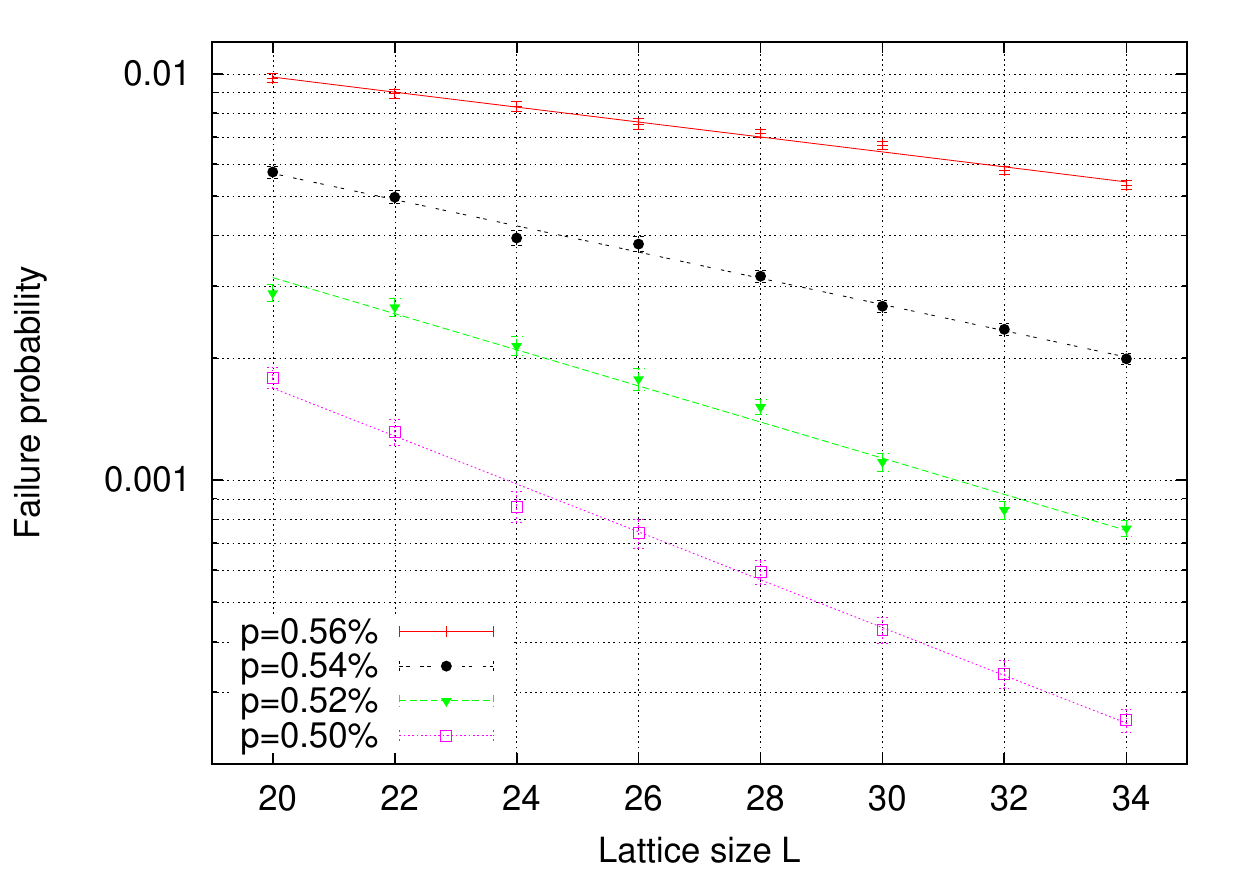}}
\caption{Scaling of the error correction failure probability for different error
rates below the threshold. Error bars represent the statistical error.
Achieving  the failure probability $10^{-6}$ at the error rate $p=0.5\%$
would require $L\approx 100$. }
\label{fig:overhead}
\end{figure}

\section{Direct $3$-qubit parity measurements}
\label{sec:toy}

In the previous section we estimated the threshold in a model where the triangle operators are measured by performing CNOT gates between the code qubits and an ancillary qubit. In this section, we will study a different model where measurements are performed directly. For some physical devices, it is possible to directly probe a multi-qubit parity operator without the need to use an ancillary qubit---the probe itself is used as a mediator that stores and accumulates the multi-qubit correlations. Direct two-qubit measurements have been realized in circuit quantum
electrodynamics~\cite{CDGN10a,FMLB09a,DCGB09a,RvKL12a}, and there are proposals \cite{Lalumiere2010} to turn these two-qubit measurements into parity measurements. These direct measurement schemes require a separate threshold analysis because they have different noise models and propagation.

We will focus on the recent proposal of DiVincenzo and Solgun~\cite{DiVincenzo12} that realizes a three-qubit parity measurement $ZZZ$ by capacitively coupling three Josephson-junction qubits to two transmission-line resonators. In the dispersive regime---where the difference between the resonant frequency of the transmission lines and the qubit transition frequency is much larger than the coupling strength---the transmission-line will pick-up a qubit-state-dependent frequency shift (Stark shift). When a near-resonant frequency probe signal is sent through one of the transmission lines, it picks up a phase that depends on the resonant frequency. Thus, the state-dependent resonant frequency shift will imprint a phase shift on the probe pulse that depends on the state of all three qubits.  With an appropriate choice of parameters (qubit-transmission line detuning, coupling strength, probe signal frequency), the probe signals sent through the two transmission lines can be measured interferometrically to reveal information only about the parity of the three qubits, all other information about the qubit state imprinted on the probe signals being erased by the interferometric measurement. One important advantage of this measurement scheme is that a single qubit can participate to two distinct parity measurements simultaneously with an appropriate arrangement of transmission lines.

The main source  of errors in this measurement is dephasing caused by the finite bandwidth of the probe pulse \cite{DiVincenzo12}. In an ideal parity measurement, the coherence between two computational basis states of the three-qubit systems $| x\ra$ and $|y\ra$ would be totally suppressed when $x$ and $y$ have different parities, and unaffected otherwise. The finite bandwidth of the probe pulse will cause dephasing between states of a given parity, and incomplete dephasing between states of distinct parity. Moreover, these errors are otherwise uniform, they do not depend, e.g., on the Hamming distance $|x-y|$ between the computational basis states. This is characteristic of a collective noise model, where multi-qubit errors are as likely as single qubit errors. In contrast, when qubits are subject to independent noise, dephasing would increase with Hamming distance between the states.

Parity measurements in the conjugate basis are required to measure the $X$-type triangle operators. These can be realized by rotating the qubits prior to sending the probe signal in the transmission line. Single qubit rotations are very fast and accurate in this architecture \cite{MGJR12a}. Nonetheless, they propagate errors, and can interchange $X$-type and $Z$-type errors. Based on these considerations, we will model noisy measurement of $X$-type and $Z$-type triangles in the following way:
\begin{itemize}
\item  A noisy $XXX$ or $ZZZ$ measurement is modeled by a perfect even/odd subspace projection, followed by
an error $(1-p) Id + p \calD$,
where $Id$ is the identity map and $\calD$ is the fully depolarizing three-qubit map
applying one of $64$ three-qubit Pauli operators with probability $1/64$ each.
\item The measurement outcome is flipped with probability p.
\end{itemize}
With this model, the syndrome extraction cycle requires only two rounds: one to measure all $X$-type triangles and one to measure all $Z$-type triangles. This implies that some qubits participate to two simultaneous measurements. As mentioned above, this is not a problem physically so far as a single qubit can be coupled to multiple transmission lines, which has already been demonstrated experimentally \cite{JRHS10a}. Moreover, two noisy parity measurements in the same basis always commute with our choice of noise model.

We have simulated fault-tolerant error correction of the subsystem toric code using direct parity measurement with the noise model described above, our results are presented on Fig.~\ref{fig:DVnoise}. Since errors are correlated in this model, we have opted for the renormalization group (RG) decoding algorithm proposed in \cite{Poulin09,DP10a1}.  Indeed, Edmonds's minimum weight matching algorithm assumes an independent noise model and consequently yields a lower threshold in the presence of correlations. The RG decoder can exactly account for some of these correlations (those that are contained within a RG unit cell). Additional correlations can be approximated by updating the error prior on each RG unit cell using a belief propagation decoder \cite{PC08a}. The RG decoding is executed using these updated error priors. Following the proposal of \cite{Bombin11}, we map the code and its (updated) noise model onto the standard toric code on which we execute RG~\footnote{This last step plays no fundamental role, it only saves us from programming a distinct RG decoder for each code: since the decoding problem of any 2D translationally-invariant stabilizer code can be mapped to that of decoding of the standard toric code~ \cite{Bombin11}, the same program can be used with different codes.}.  While Refs.~ \cite{Poulin09,DP10a1} assumed noiseless syndrome measurements, the decoding algorithm can easily be extended to noisy syndrome by renormalizing a 3D lattice \cite{DP12a}.

\begin{figure}[h]
\centerline{\includegraphics[height=6cm]{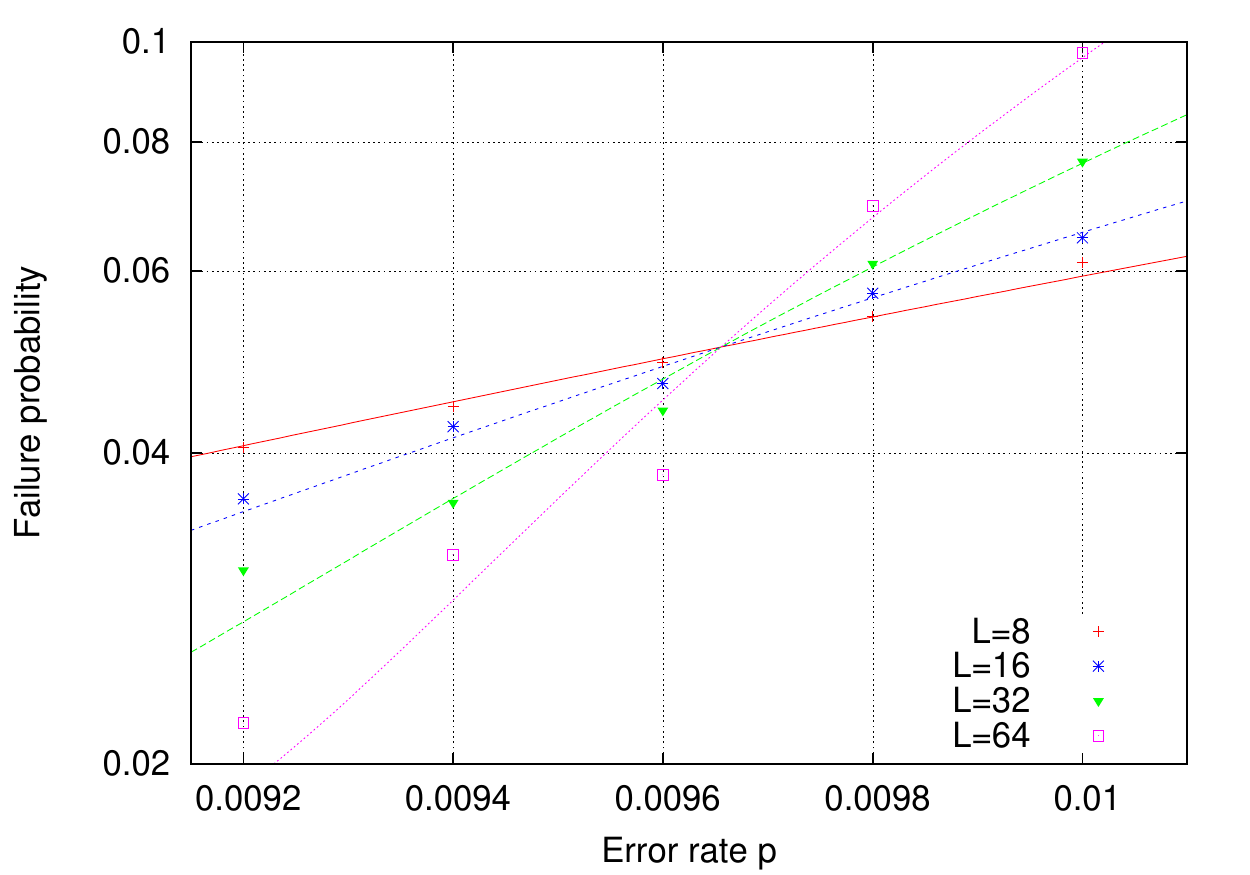}}
\caption{Decoding error probability as a function of the physical error rate $p$ for various lattice sizes $L$. The lines are obtained using the fitting functions $a + b (p - p_{th}) L^{1/\nu} + c [(p - p_{th}) L^{1/\nu}]^2$, which produces a threshold $p_{th} \approx 0.968 \%$ and a critical exponent $\nu \approx 1.36$.}
\label{fig:DVnoise}
\end{figure}

The results shown on Fig.~\ref{fig:DVnoise} indicate a threshold value of roughly 0.97\%. This value should be seen as  lower bound to the true threshold of this code, which may well be above 1\%. Indeed, the results presented here were obtained from a unit cell of dimension $2\times 2 \times 1$: two of the three space-time dimensions are renormalized at each iteration, and by rotating the unit cell at each iteration we obtain a renormalization of the entire space-time lattice by a factor 4 after 3 RG rounds.  Based on our experiments, larger RG unit cells produce higher thresholds because they make use of more correlations existing in the noise model. The decoding complexity scales exponentially with the unit cell size however, so in practice we are limited to relatively small cells.
Furthermore, while the simulated syndrome extraction protocol involved measurements of both X-type and Z-type triangles, error correction has been performed only for X-type errors. Applying the same error correction algorithm independently to X-type and Z-type errors would result in the same error threshold~\footnote{Since the subsystem toric code and the error model are symmetric under exchange of Pauli $X$ and $Z$, the error threshold must be the same
for the models with only $X$ and only $Z$  errors. Correcting each type of errors independently for the full error model can increase the decoding failure probability at most by a factor of two
compared with the case of $X$ errors only (use the union bound).}. We expect that more sophisticated decoders
taking into account correlations between $X$ and $Z$ errors could achieve
higher error thresholds.

Finally, we note that direct parity measurements of weight-four operators such as $ZZZZ$ and $XXXX$ can be realized similarly \cite{DiVincenzo12}. This could be used to implement Kitaev's toric code. However, it is not reasonable to assume that the noise rate $p$ is independent of the weight $w$ of the operator being measured. Thus, one needs to work out this dependence $p(w)$ from physical considerations before comparing thresholds of different codes obtained from direct parity measurements.

\section{Conclusion}

We have presented a subsystem version of Kitaev's surface code. The main features of our code is that it requires only 3-qubit parity measurements and its stabilizer generators have weight 6. Minimizing the weight of the parity measurements is helpful as it simplifies the measurement procedure, while minimizing the weight of the stabilizer generators is also desirable since it makes syndromes more reliable. In contrast to our code, the standard toric code requires weight 4 parity measurements and has weight 4 stabilizer generators. The subsystem color codes require only weight 2 parity measurements, but 
have stabilizer generators of weight up to 18. Thus, based only on these considerations, it is not clear how the threshold of these various codes should compare.

Our numerics show that in the circuit based model, our code has a threshold (0.6\%)
which is almost an order of magnitude larger than the one of the color code (0.08\%) \cite{LAR11a}, and a bit more than half that of the standard toric
code (0.9\%) \cite{WFH:highthresh}. Motivated by the recent work of DiVincenzo and Solgun~\cite{DiVincenzo12}, we have also considered a setting where parity measurements can be implemented directly and found a threshold of 0.97\%. This value cannot be compared directly to the thresholds reported above since it is based on a substantially different noise model, and additional physical considerations must be taken into account before comparing.

We have shown that the new subsystem toric code gives rise to an exactly
solvable spin Hamiltonian with $3$-qubit interactions and
topologically ordered ground state which is
locally equivalent to the standard toric code.

\section{Acknowledgements}

We would like to  thank David DiVincenzo for helpful discussions
and for drawing our attention to Ref.~\cite{DiVincenzo12}.
This work was supported by Intelligence Advanced Research Projects Activity (IARPA) via Department of Interior National Business Center contract D11PC20167.
S.B. was partially supported by DARPA QUEST program under contract number HR0011-09-C-0047.
The U.S. Government is authorized to reproduce and distribute reprints for Governmental purposes notwithstanding any copyright annotation thereon. Computational resources were provided by Calcul Qu\'ebec
 and by IBM Blue Gene Watson supercomputing center.
 Disclaimer: The views and conclusions contained herein are those of the authors and should not be interpreted as necessarily representing the official policies or endorsements, either expressed or implied, of IARPA, DoI/NBC, or the U.S. Government.


\end{document}